\documentclass[a4paper,UKenglish,cleveref, autoref, thm-restate,]{lipics-v2021}
\pdfoutput=1

\title{On Homogeneous Models of Fluted Languages} 

\titlerunning{On Homogeneous Models of Fluted Languages} 

\usepackage[inline]{enumitem}
\usepackage{stfloats} 
\usepackage{array}
\usepackage{makecell}
\newcolumntype{?}{!{\vrule width 1.1pt}}
\usepackage{xparse}

\NewDocumentCommand{\N}{}{\ensuremath{\mathbb{N}}} 

\NewDocumentCommand{\sizeof}{m}{||#1||}
\NewDocumentCommand{\fA}{}{\ensuremath{\mathfrak{A}}} 
\NewDocumentCommand{\fB}{}{\ensuremath{\mathfrak{B}}} 
\NewDocumentCommand{\fG}{}{\ensuremath{\mathfrak{G}}} 

\NewDocumentCommand{\AF}{}{\ensuremath{\mathcal{AF}}} 

\NewDocumentCommand{\defctl}{ O{} O{\ensuremath{\tau}} }{\ensuremath{\langle #2, L #1, R #1\rangle}}
\NewDocumentCommand{\defctlzero}{ O{\ensuremath{\tau}} }{\ensuremath{\langle #1, \emptyset, \emptyset \rangle}}

\NewDocumentCommand{\defmsc}{ O{} O{\fA} O{\fB} }{\ensuremath{\langle \Lambda_{#2} #1, \Lambda_{#3} #1, L #1, R #1 \rangle}}
\NewDocumentCommand{\defmsczero}{ O{} O{\fA} O{\fB} }{\ensuremath{\langle \Lambda_{#2} #1, \Lambda_{#3} #1, \emptyset, \emptyset \rangle}}

\NewDocumentCommand{\defbhv}{ O{} O{\ensuremath{\mu}} O{\ensuremath{\lambda}} }{\ensuremath{\langle #2, #3, L #1, R #1, w #1\rangle}}
\NewDocumentCommand{\defbhvzero}{ O{\ensuremath{\mu}} O{\ensuremath{\lambda}} }{\ensuremath{\langle #1, #2, \emptyset, \emptyset, \emptyset \rangle}}

\NewDocumentCommand{\wordctl}{ O{i} }{$#1$-constellation}
\NewDocumentCommand{\wordctls}{ O{i} }{$#1$-constellations}

\NewDocumentCommand{\ftp}{}{\ensuremath{ \mathsf{ftp} }}

\NewDocumentCommand{\FTP}{}{\ensuremath{ \mathsf{FTP} }}

\NewDocumentCommand{\fpr}{}{\ensuremath{ \mathsf{fpr} }}

\NewDocumentCommand{\FL}{}{\ensuremath{\mathcal{FL}}}

\NewDocumentCommand{\FLC}{}{\ensuremath{\mathcal{FLC}}}

\NewDocumentCommand{\FLPC}{}{\ensuremath{\mathcal{FLPC}}}

\NewDocumentCommand{\restr}{}{\ensuremath{\!\!\restriction\!\!}}

\NewDocumentCommand{\psA}{ O{} O{} O{} }{\ensuremath{(\fA_{#1}, #2 \bar{a}_{#1} #3)}}

\NewDocumentCommand{\psB}{ O{} O{} O{} }{\ensuremath{(\fB_{#1}, #2 \bar{b}_{#1} #3)}}

\NewDocumentCommand{\sig}{}{\text{sig}}

\NewDocumentCommand{\modInts}{ O{p} }{\ensuremath{\mathbb{Z} / #1 \mathbb{Z}}}

\newcommand{\NN}{\ensuremath{\N{\times}\N}} 
\newcommand{\NNs}{\ensuremath{(\NN)}} 

\author{Daumantas Kojelis}
{Department of Computer Science, University of Manchester, UK
  \and \url{https://daumantaskojelis.github.io/}
}
{daumantas.kojelis@manchester.ac.uk}
{https://orcid.org/0000-0002-1632-9498}
{}

\authorrunning{D. Kojelis} 

\Copyright{Daumantas Kojelis} 



\category{} 

\relatedversion{} 

\supplement{}

\funding{This work is supported by the NCN grant 2018/31/B/ST6/03662.}

\acknowledgements{The author would like to thank Dr. Ian Pratt-Hartmann and Dr. Bartosz Jan Bednarczyk for their valuable feedback.}

\nolinenumbers 

\hideLIPIcs  

\EventEditors{J\"{o}rg Endrullis and Sylvain Schmitz}
\EventNoEds{2}
\EventLongTitle{33rd EACSL Annual Conference on Computer Science Logic (CSL 2025)}
\EventShortTitle{CSL 2025}
\EventAcronym{CSL}
\EventYear{2025}
\EventDate{February 10--14, 2025}
\EventLocation{Amsterdam, Netherlands}
\EventLogo{}
\SeriesVolume{326}
\ArticleNo{9}

\begin{document}

\maketitle

\begin{abstract}
    We study the fluted fragment of first-order logic which is often viewed as a multi-variable non-guarded
    extension to various systems of description logics lacking role-inverses.
    In this paper we show that satisfiable fluted sentences (even under reasonable extensions) admit
    special kinds of ``nice'' models which we call globally/locally homogeneous.
    Homogeneous models allow us to simplify methods for analysing fluted logics with counting quantifiers and
    establish a novel result for the decidability of the (finite) satisfiability problem for the fluted fragment with periodic counting.
    More specifically, we will show that the (finite) satisfiability problem for the language is \textsc{Tower}-complete.
    If only two variable are used, computational complexity drops to \textsc{NExpTime}-completeness.
    We supplement our findings by showing that generalisations of fluted logics, such as the adjacent fragment,
    have finite and general satisfiability problems which are, respectively, $\Sigma^0_1$- and $\Pi^0_1$-complete.
    Additionally, satisfiability becomes $\Sigma^1_1$-complete if periodic counting quantifiers are permitted.
\end{abstract}

\section{Introduction}

The \textit{fluted fragment} (denoted $\FL$) is a fragment of first-order logic in which,
roughly put, variables appear in predicates following the order in which they were quantified.
For illustrative purposes, we translate the sentence ``Every conductor nominates their favorite soloist to play at every concert'' into
this language as follows:
\begin{equation} \label{eq:intro_example}
    \forall x_1 \Big(cond(x_1) \rightarrow
    \exists x_2 \big(solo(x_2) \wedge fav(x_1, x_2) \wedge
    \forall x_3 (conc(x_3) \to nom(x_1, x_2, x_3) ) \big) \Big).
\end{equation}
As a non-example, the sentences axiomatising transitivity, symmetry and reflexivity of a relation are not in the fluted fragment.

The fluted fragment is a member of \textit{argument-sequence logics} -- a family of decidable (in terms of satisfiability) fragments of first-order logic
which also includes the \textit{ordered} \cite{herzig90, Jaakkola21}, \textit{forward} \cite{Bednarczyk21}
and \textit{adjacent} \cite{bki23} fragments.
The fluted fragment in particular is decidable in terms of satisfiability even in the presence of
counting quantifiers \cite{ph21} or a distinguished transitive relation \cite{pl_lt19}. Surprisingly, the satisfiability problem for $\FL$
under a combination of the two not only retains decidability but also has the finite model property~\cite{ph_lt23}.
We refer the reader to \cite{adj:p-ht22} for a survey.

In this paper we will mostly be concerned with what we call the \textit{fluted fragment with periodic counting} (denoted $\FLPC$).
We remark that periodic counting quantifiers generalise standard (threshold) counting quantifiers
which have been an object of intensive study as an extension for the fluted fragment in the past few years \cite{ph21, ph_lt23}. 
Under this new formalism, we are allowed to write formulas requesting an even number of existential witnesses.
As an example, we can express sentences as ``Every orchestra hires an even number of people to play first violin'' in our language (see (\ref{eq:prereq_example})).

The origins of \textit{flutedness} trace back to the works of W. V. Quine \cite{quine69}. It is, however, the definition given by W. C. Purdy (in \cite{purdy96})
that has become widespread and will be the one we use. The popularity of Purdy's idea of flutedness is not without cause,
at least when keeping the field of description logics in mind.
Indeed, after a routine translation,
formulas of the description logic $\mathcal{ALC}$ are contained
in the two-variable sub-fragment of $\FLPC$. This is even the case when $\mathcal{ALC}$ is augmented with
role hierarchies, nominals and/or cardinality restrictions (possibly with modulo operations). We refer the reader to \cite{hsg04} for more details.
In terms of expressive power, $\FLPC$ closely parallels $\mathcal{ALCSCC}$ --
a new formalism with counting constrains expressible in quantifier-free Boolean algebra with Presburger arithmetic (see \cite{QFBAPA, ALCSCC}).
Thus, noting that the guarded fragment with at least three variables becomes undecidable
under counting extensions \cite{Gradel99}, and that the
guarded fluted fragment has ``nice'' model theoretic properties such as \textit{Craig interpolation} \cite{BednarczykJ22},
fluted languages emerge as perfect candidates
for generalising description logics in a multi-variable context.

\begin{table*}[b]
    \begin{tabular}{|c?c|c||c|c|}
    \hline
  
                & $\FL^2$                                           & $\FL^\ell$                                          & $\AF^3$                                                      & $\AF^{k}$                                                                             \\
    \Xhline{3\arrayrulewidth}
    
    standard& \textsc{NExp}-c \cite{GradelKV97}                 & $(\ell{-}2)$-\textsc{NExp} \cite{phst19}              & \textsc{NExp}-c \cite{bki23}                                 & $(k{-}2)$-\textsc{NExp} \cite{bki23}                                                    \\
    \hline
  
    counting&  \textsc{NExp}-c \cite{Pratt-Hartmann10}          & $(\ell{-}1)$-\textsc{NExp} \cite{ph21}                & $\Sigma^0_1$-c/$\Delta^0_1$ Th \ref{th_FL_rev_fin_undec}/claim            & $\Sigma^0_1$-c/$\Pi^0_1$-c Th \ref{th_FL_rev_fin_undec}/\ref{th_FLCrev_FLPCrev_gen_undec}    \\
    \hline
  
    periodic& \textsc{NExp}-c Th~\ref{th:2_var_NExp}            & $(\ell{-}1)$-\textsc{NExp} Th~\ref{th:more_vars_time} & $\Sigma^0_1$-c/$\Sigma^0_1$-h Th \ref{th_FL_rev_fin_undec}  & $\Sigma^0_1$-c/$\Sigma^1_1$-c Th \ref{th_FL_rev_fin_undec}/\ref{th_FLCrev_FLPCrev_gen_undec} \\
    \hline
  
    \end{tabular}
    \mbox{}
    %
    \caption{Complexity of finite (left-hand side of ``/'') and general (right-hand side of ``/'') satisfiability problems for languages (in the top row) under quantifier extensions (on the left-most column).
    All complexity classes are in regard to time. $\mathcal{C}$-c ($\mathcal{C}$-h) stands for complete (hard). Here $k \geq 4$ and $\ell \geq 3$.}
    %
    \label{table_complexity}
  \end{table*}

In Sections \ref{sec:2vars} and \ref{sec:more} we establish that classes of models of satisfiable $\FLPC$-sentences always contain a ``nice'' structure in which
elements behave (in a sense that we will make clear) \textit{homogeneously}.
Utilising this behaviour we will show that the fluted fragment extended with periodic counting quantifiers
has a decidable satisfiability problem. Intriguingly, even though periodic counting quantifiers generalise standard counting quantifiers,
our methodology allows us to avoid \textit{Presburger quantification}, which was required to
establish decidability of satisfiability for $\FL$ with standard counting \cite{ph21}.

In section \ref{sec_undec} we show that the satisfiability problems for the fluted fragment with counting extensions become undecidable
when minimal syntactic relaxations are allowed. More precisely, the section will culminate with a result showing that
the finite satisfiability problem for the 3-variable adjacent fragment with counting is $\Sigma^0_1$-complete.
Additionally, the general satisfiability problem will be shown to be $\Pi^0_1$-complete when 4 variables are used,
and $\Sigma^1_1$-complete if periodic counting is allowed. Denoting the adjacent fragment as $\mathcal{AF}$, 
we provide a brief survey of complexity and undecidability standings in Table~\ref{table_complexity}.

The work in this paper is closely related to \cite{BenediktKT20} in which decidability of satisfiability is established
for the two-variable fragment with periodic counting (denoted $\mathcal{FO}^2_{\text{Pres}}$) but without a sharp complexity-theoretic bound.
Our homogeneity conditions, which stem from lack of inverse relations in fluted logics, allow us
to establish \textsc{NExpTime}-completeness for both the finite and general satisfiability problems of $\FLPC^2$.
\pagebreak

\section{Preliminaries} \label{sec:prereq}

We use $\mathbb{N}$ to denote the set of integers $\{ 0, 1, 2, \dots \}$, and
$\mathbb{N}^*$ to denote $\N$ along with the first infinite cardinal; i.e.
$\mathbb{N}^* = \mathbb{N} \cup \{ \aleph_0 \}$.
By picking some $n, p \in \mathbb{N}$
we write $n^{+p}$ for the \textit{linear set} $\{ n + ip \mid i \in \mathbb{N} \}$.
In the extended integers $\mathbb{N}^*$, the cardinal $\aleph_0$ is the maximum element under the canonical ordering ``$<$''
and
\begin{itemize}
    \item $0 \cdot \aleph_0 = \aleph_0 \cdot 0 = 0$;
    \item $n + \aleph_0 = \aleph_0 + n = \aleph_0$ for all $n \in \mathbb{N}^*$; and
    \item $n \cdot \aleph_0 = \aleph_0 \cdot n = \aleph_0$ for all $n \in \mathbb{N}^* \setminus \{ 0 \}$.
\end{itemize}
A \textit{linear Diophantine inequation} is an expression of the form
$
    a_1 v_1 + \dots + a_n v_n + b \bowtie c_1 v_1 + \dots + c_n v_n + d,
$
where $(a_i)_{i = 1}^{n}, b, (c_i)_{i = 1}^{n}, d$ are constant values taken form
$\mathbb{N}^*$, $\bar{v} = v_1, \dots v_n$ is a vector of variables, and ``$\bowtie$'' is any of the relations ``$=, \neq, \leq, <, \geq, >$''
(each interpreted as one would assume). 
It is known that when the cardinal $\aleph_0$ is disallowed, a solution for a set of such inequations may be found in \textsc{NPTime} \cite{ilp_nptime}.
The picture does not change when $\aleph_0$ is permitted as a solution and/or constant.
Indeed, we may reduce the problem of finding a solution over $\mathbb{N}^*$ to that of finding it over $\mathbb{N}$ as follows.
First guess which variables should be mapped to $\aleph_0$ and which should have a finite value.
Then, check that each inequation featuring a variable assigned $\aleph_0$ holds and discard them.
What will be left is a system of inequations with constants in $\mathbb{N}$ and in variables assumed to be finite.
See \cite[Ch~7.4]{PrattHartmann23} for greater detail.
We will allow systems of inequations to contain disjunctions.

By a word $\bar{a} \in A^n$ we mean a tuple $\bar{a} = a_1 \cdots a_n$, where $a_i \in A$ for each $i$, $1 \leq i \leq n$.
In case $n = 1$, we often write $a$ instead of $\bar{a}$.
By $\bar{a}^{-1}$ we mean the reversal of $\bar{a}$; i.e. $\bar{a}^{-1} = a_n \cdots a_1$.
If $\bar{a}$ and $\bar{b}$ are words we write $\bar{a}\bar{b}$ for the concatenation of the two.
We use the terms ``tuple'' and ``word'' interchangeably. 

Now, take some structure $\fA$ and an $i$-tuple $\bar{a}$ of elements from $A$.
Suppose $B = \{ b \in A \mid \fA, \bar{a}b \models \varphi \}$ for some first-order formula $\varphi(x_1, \dots, x_{i{+}1})$.
Fixing $n, p \in \mathbb{N}$
we extend the syntax of first-order logic
with \textit{(threshold) counting quantifiers} $\exists_{[\geq n]}$ and \textit{periodic counting quantifier} $\exists_{[n^{+p}]}$.
Semantically, $\fA, \bar{a} \models \exists_{[\geq n]} x_{i{+}1} \varphi$ if and only if $|B| \geq n$.
Similarly, $\fA, \bar{a} \models \exists_{[n^{+p}]} x_{i{+}1} \varphi$
if and only if $|B| \in n^{+p}$.
We refrain from further generalisation to \textit{ultimately periodic counting quantifier}
$\exists_{[n_1^{+p_1} \cup \dots \cup n_k^{+p_k}]}$ (as in \cite{BenediktKT20})
as they can be expressed as a disjunction
of formulas using periodic counting quantifiers $\exists_{[n_1^{+p_1}]} x_{i{+}1} \varphi \vee \dots \vee \exists_{[n_k^{+p_k}]} x_{i{+}1} \varphi$.
Thus, a sentence such as
``Every orchestra hires an even number of people to play first violin'' may be written in a language with periodic counting:
\begin{equation} \label{eq:prereq_example}
    \forall x_1 \Big( orch(x_1) \to \exists_{[0^{+2}]} x_2 \big( pers(x_2) \wedge
    \exists x_3 (1st\_viol(x_3) \wedge
     hires\_to\_play(x_1, x_2, x_3)
        )
    \big)
    \Big).
\end{equation}

Formally, the fluted fragment with periodic counting is the union of sets of formulas
$\FLPC^{[\ell]}$ defined by simultaneous induction as follows:
\begin{enumerate}[label=(\roman*)]
    \item any atom $r(x_k, \dots, x_\ell)$, where $x_k, \dots, x_\ell$ is a contiguous subsequence of $x_1, x_2, \dots$ and $r$
    is a predicate of arity $\ell {-} k {+} 1$, is in $\FLPC^{[\ell]}$;
    \item $\FLPC^{[\ell]}$ is closed under Boolean combinations;
    \item if $\varphi \in \FLPC^{[\ell{+}1]}$, then $\exists_{[n^{+p}]} x_{\ell{+}1} \varphi$ is in $\FLPC^{[\ell]}$
    for every $n, p \in \mathbb{N}$.
\end{enumerate}
We write $\FLPC = \bigcup_{\ell \geq 0} \FLPC^{[\ell]}$ for the set of all fluted formulas with periodic counting and define
the $\ell$-variable fluted fragment with periodic counting to be the set $\FLPC^\ell := \FLPC \cap \mathcal{FO}^\ell$.
(Here $\mathcal{FO}^\ell$ is the set of first-order formulas that do not use more than $\ell$ variables).
We will implicitly restrict attention to signatures $\sigma \cup \{ = \}$ which feature no function and/or constant symbols, and where
``$=$'' is always interpreted as the canonical equality relation. Lastly, we use $\forall x \varphi$ interchangeably with $\exists_{[0^{+0}]}x\;\neg \varphi$
whenever convenient. 

Variables in fluted logics convey no meaningful information.
Indeed, since the arity of every predicate in $\sigma \cup \{ = \}$ is fixed and each atom features a suffix of the variable quantification order,
we may employ what we call \textit{variable-free notation}. As an example, formulas (\ref{eq:intro_example}) and (\ref{eq:prereq_example})
may be (respectively) written as 
\begin{align}
    &\forall \Big(cond \rightarrow
    \exists \big(solo \wedge fav \wedge \mbox{}
    \forall (conc \to nom ) \big) \Big), \\
    &\forall \Big( orch \to
    \exists_{[0^{+2}]} \big( person \wedge \exists (1st\_viol \wedge
     hires\_to\_play )
        \big)
    \Big)
\end{align}
without ambiguity (up to a shift of variable indices).
Writing $\forall^\ell$ for $\forall x_1 \cdots \forall x_{\ell}$, we will employ variable-free notation extensively
throughout subsequent sections.

Fix a first-order formula $\varphi$ with periodic counting quantifiers. We
assume that numeric values are encoded in binary and write $||\varphi||$ for the number of symbols used in $\varphi$.
We point out that the signature of $\varphi$ (which we write as $\sig(\varphi)$ for brevity)
is no larger than~$||\varphi||$.

Now, let $\varphi$ be a formula in $\FLPC^{\ell{+}1}$. We say that $\varphi$ is in \textit{normal-form}
if it takes the following shape
\begin{equation} \label{eq:nmf}
    \bigwedge_{r \in R} \forall^\ell \Big( \alpha_r \to \exists_{[n_r^{+p_r}]} \gamma_r \Big) \wedge
    \bigwedge_{t \in T} \forall^\ell \Big( \beta_t \to \neg \exists_{[n_t^{+p_t}]} \delta_t \Big),
\end{equation}
where $R, T$ are sets of indices, each $\alpha_r, \beta_t$ is a quantifier-free formula in $\ell$ variables, each $\gamma_r, \delta_t$ is a
quantifier-free formula in $\ell{+}1$ variables, and each $n_r^{+p_r}, n_t^{+p_t}$ is a linear set.
Using standard rewriting techniques we prove the following in Appendix \ref{app:nmf}.
\begin{lemma} \label{lma:nmf}
    Suppose $\varphi$ is an $\FLPC^{\ell{+}1}$-sentence. Then, we may compute, in polynomial time, an equisatisfiable
    normal-form $\FLPC^{\ell{+}1}$-sentence $\psi$.
\end{lemma}
%
Notice that the negation before the periodic counting quantifier in the second conjunct of (\ref{eq:nmf}) is not moved-inwards.
This is done deliberately so as to avoid computing complements of linear sets, which may be of exponential size as a function of $\sizeof{\varphi}$.

Now, let $\fA$ be a structure over a finite signature $\sigma \cup {\{=\}}$.
We say that an atom is $(\ell{+}1)$-fluted if it features a suffix of $x_1, \dots, x_{\ell{+}1}$ as arguments.
The fluted atomic $(\ell{+}1)$-type $\tau$ of $a\bar{b}c \in A^{\ell{+}1}$
(denoted $\ftp^\fA_{\ell{+}1}(a\bar{b}c)$) is then
the set of $(\ell{+}1)$-fluted (possibly negated) atoms over $\sigma \cup \{ = \}$ with arity no greater than
$\ell{+}1$ that are satisfied by $a\bar{b}c$ in~$\fA$.
We invite the reader to view the tuple $a\bar{b}$ as \textit{emitting} $\tau$, and $\bar{b}c$ as \textit{absorbing} $\tau$.
Write $\tau \restriction_{[2, \ell{+}1]}$ for the fluted $\ell$-type $\pi$ obtained by
deleting entries in $\tau$ of arity greater than $\ell$ and decrementing variable indices by 1.
We say that a fluted $\ell$-type $\pi$ is an endpoint of a fluted $(\ell{+}1)$-type $\tau$ if $\tau \restriction_{[2, \ell{+}1]} = \pi$.
We define $\FTP^\sigma_{\ell{+}1}$ to be the set of all fluted $(\ell{+}1)$-types over $\sigma \cup \{=\}$.

The fluted $\ell$-profile of $a\bar{b} \in A^\ell$ (denoted $\fpr^\fA_\ell(a\bar{b})$) is a function $\rho$ mapping $\tau \in \FTP^\sigma_{\ell{+}1}$
to the number of times $a\bar{b}$ emits $\tau$. Formally,
$
    \rho(\tau) = | \{ c \in A \mid \ftp^\fA_{i{+}1}(a\bar{b}c) = \tau \} |.
$
If $\varphi$ is a quantifier-free $\FLPC^{\ell{+}1}$-formula, we write
$\rho \models \exists_{[n^{+p}]} \varphi$ just in case $\sum_{\tau \in \FTP^\sigma_{\ell{+}1}}^{\tau \models \varphi} \rho(\tau) \in n^{+p}$.
Clearly, $\rho \models \exists_{[n^{+p}]} \varphi$ if and only if $\fA, a\bar{b} \models \exists_{[n^{+p}]} \varphi$.

In the sequel we will assume that all structures are countable. This comes with no loss of generality as $\FLPC$
is subsumed by the countable fragment of infinitary logic $\mathcal{L}_{\omega_1, \omega}$; a language for which 
the downward L\"{o}wenheim-Skolem Theorem holds (folklore, see \cite[p.~69]{infi_logic}).
Note that, in $\FLPC$, the finite model property fails as $\neg \exists_{[0^{{+}1}]} \top$ is an axiom of~infinity.
\section{The Two-Variable Sub-fragment} \label{sec:2vars}

In this section we restrict attention to the two-variable fluted fragment with periodic counting.
To achieve decidability of (finite) satisfiability we first specify what kind of ``nice''
models we will be looking for.
Take any $\sigma$-structure $\fA$ and $\pi \in \FTP^\sigma_1$.
For convenience, we write $A_\pi$ for the set of all elements $a \in A$ with $\ftp^\fA(a) = \pi$.
We say that $\pi$ is \textit{globally homogeneous} in $\fA$ if $\fpr^\fA(a) = \fpr^\fA(b)$ for each $a, b \in A_\pi$.
That is to say, $\pi$ is globally homogeneous in $\fA$ if, for each $\tau \in \FTP^\sigma_2$, the number of fluted 2-types $\tau$ emitted 
is the same for each element in $A_\pi$.
The structure $\fA$ is \textit{globally homogeneous} if each $\pi \in \FTP^\sigma_1$ is globally homogeneous in $\fA$.

For the remainder of the section fix some normal-form $\FLPC^2$-sentence $\varphi$.
We claim that if $\varphi$ is a satisfiable, then it is satisfiable in a globally homogeneous model.
To see this, take some structure $\fA \models \varphi$ and a pair of elements $cd \in A^2$ that realised the fluted 2-type $\tau$ in $\fA$.
Since $\tau$ consists of fluted formulas, it does not feature atoms of the form $p(x_1)$ and $p(x_2, x_1)$.
Referencing $\tau$ only, we lack information to deduce what formulas are satisfied by the pair $dc$ in $\fA$.
On the other hand, $\fA, cd \models \psi$ if and only if $\tau \models \psi$ for each quantifier-free $\FLPC^2$-formula $\psi$.
Thus, if we were to in any way alter the fluted 2-type of $dc$ in $\fA$, the set of quantifier-free $\FLPC^2$-formulas satisfied by $cd$ in $\fA$
would not change.

Taking a step back, take any $\pi \in \FTP^\sigma_1$ and recall that $A_\pi \subseteq A$ is the set of all elements realising the 1-type $\pi$ in $\fA$.
We redefine 2-types emitted by elements of $A_\pi$ in $\fA$ in such a way that makes $\pi$ globally homogeneous in the rewiring of $\fA$.
%
Let us fix any $a \in A_\pi$ and write $\rho = \fpr^\fA(a)$. The element $a$ and profile $\rho$ picked will serve as an ``example''
of how the rest of $A_\pi$ should form fluted 2-types with other elements of the model.
Taking any $b \in A_\pi \setminus \{ a \}$ and $c \in A \setminus \{ a, b \}$ we see that it is impossible to prohibit
$b$ from emitting the fluted 2-type $\ftp^\fA(ac)$ to $c$ via any normal-form $\FLPC^2$-formula.
We thus allow $b$ to impersonate $a$ in $\fA$ by rewiring the fluted 2-type $\ftp^\fA(bc)$ to be $\ftp^\fA(ac)$
for each $c \in A \setminus \{ a, b \}$ and, additionally, by resetting $\ftp^\fA(ba)$ to be the 2-type $\ftp^\fA(ab)$ and $\ftp^\fA(bb)$ to be $\ftp^\fA(aa)$.
Clearly, only fluted 2-types emitted by $b$ were reconsidered in this procedure, thus pairs in $(A \setminus \{ b \}) \times A$
satisfy the same quantifier-free $\FLPC^2$-formulas as before.
To see that $\fA$ still models $\varphi$ we need only show that $b$ does not violate $\alpha_r \to \exists_{[n_r^{+p_r}]} \gamma_r$
and $\beta_t \to \neg \exists_{[n_t^{+p_t}]} \delta_t$ for each $r \in R$ and $t \in T$.
By our rewiring procedure, we have that $\fpr^\fA(a) = \rho = \fpr^\fA(b)$.
Thus,
\begin{equation*}
    \fA, a \models \exists_{[n^{+p}]} \psi \iff \rho \models \exists_{[n^{+p}]} \psi \iff \fA, b \models \exists_{[n^{+p}]} \psi
        \qquad \text{for quantifier-free $\psi \in \FLPC^2$.}
\end{equation*}
Having already argued that $\fA, a \models \alpha_r \to \exists_{[n_r^{+p_r}]} \gamma_r$ and 
$\fA, a \models \beta_t \to \neg \exists_{[n_t^{+p_t}]} \delta_t$ for each $r \in R$ and $t \in T$ we therefore conclude that
$\fA \models \varphi$.

Since the only 2-types redefined are those emitted by $b$, we can run this construction in parallel
for each element in $A_\pi \setminus \{ a \}$.
Clearly, this renders $\pi$ globally homogeneous in the rewired model.
Since elements in $A \setminus A_\pi$ are left untouched by our rewiring, 
we can repeat this procedure for each $\pi \in \FTP^\sigma_1$
thus proving the following:
\begin{lemma} \label{cor:glob_homo}
    Suppose $\varphi$ is a satisfiable normal-form $\FLPC^2$-sentence. Then, $\varphi$ is satisfiable in a globally homogeneous model.
\end{lemma}

In globally homogeneous structures elements realising the same fluted 1-type are, in a sense, stripped away of their individuality
as they all realise the same fluted 1-profile. When the globally homogeneous structure $\fA$ is clear from context,
we can unambiguously write $\rho_\pi$ for the fluted 1-profile
realised by each element of $A_\pi$ (for $\pi \in \FTP^\sigma_1$).

When considering the (finite) satisfiability problem for normal-form $\FLPC^2$-sentences such as $\varphi$,
we will confine ourselves to the search of globally homogeneous models. 
More precisely, we will produce a system of linear Diophantine inequations $\Psi$ that has a solution over $\N^*$ if and only if
$\varphi$ is satisfiable in a globally homogeneous model.
For this purpose, let
$(x_\pi)_{\pi \in \FTP^\sigma_1}$,
$(y_{\pi, \tau})_{\pi \in \FTP^\sigma_1}^{\tau \in \FTP^\sigma_2}$,
$(i_{\pi, r})_{\pi \in \FTP^\sigma_1}^{r \in R}$, and
$(j_{\pi, t})_{\pi \in \FTP^\sigma_1}^{t \in T}$
be sequences of variables.
Intuitively, the value assigned to $x_\pi$ will represent the number of elements realising the fluted 1-type~$\pi$;
$y_{\pi, \tau}$ -- the number times the 2-type $\tau$ is emitted by an element realising $\pi$; and
with $i_{\pi, r}$ and $j_{\pi, t}$ acting as periodic counters for elements realising $\pi$ when considering linear sets $n_r^{+p_r}$ and~$n_t^{+p_t}$.
To be more precise, when given a satisfying assignment for $\Psi$, we will build a globally homogeneous $\fA \models \varphi$ with
\begin{equation}\label{FLPC2_assignments_2}
    |A_\pi| = x_\pi \text{ and } \rho_\pi(\tau) = y_{\pi, \tau} \qquad\qquad\qquad\qquad\qquad \text{ for each $\pi \in \FTP^\sigma_1$ and $\tau \in \FTP^\sigma_2$.}
\end{equation}
On the other hand, 
we will have that any globally homogeneous model $\fA \models \varphi$ gives rise to a solution
of $\Psi$ with the following assignments for all $\pi \in \FTP^\sigma_1$, $\tau \in \FTP^\sigma_2$, $r \in R$ and $t \in T$:\footnote{
    In case $p_r$ (resp. $p_t$) is $0$, we allow $i_{\pi, r}$ (resp. $j_{\pi, t}$) to take any integer value.}
\begin{equation}\label{FLPC2_assignments}
    x_\pi {:=} |A_\pi|; \;
    y_{\pi, \tau} {:=} \rho_\pi(\tau); \;
    i_{\pi, r} {:=} \Big( \sum_{\tau' \in \FTP^\sigma_2}^{\tau' \models \gamma_r} \hspace{-2mm} \rho_\pi(\tau') - n_r \Big) / p_r; \;
    j_{\pi, t} {:=} \bigg\lfloor \Big( \sum_{\tau' \in \FTP^\sigma_2}^{\tau' \models \delta_t} \hspace{-2mm} \rho_\pi(\tau') - n_t \Big) / p_t \bigg\rfloor.
\end{equation}
We will proceed first by showing that the latter assignment satisfies our (yet to be defined) system of inequations $\Psi := \Psi_1 \cup \cdots \cup \Psi_6$.

When considering any model $\fA \models \varphi$ we have that the domain $A = \bigcup_{\pi \in \FTP^\sigma_1} A_\pi$ is non-empty.
The following is thus trivially satisfied by the assignment $x_\pi := A_\pi$:
\begin{equation}
    \Big\{ \sum_{\pi \in \FTP_1} x_\pi \geq 1 \Big\}. \tag{$\Psi_1$}\label{eq_FLPC2_Psi1}
\end{equation}
Additionally, picking any element $a \in A_\pi$ (for any $\pi \in \FTP_1^\sigma$)
and picking any $\pi' \in \FTP^\sigma_{\pi'}$, we have that the number of fluted 2-types emitted by $a$ to $A_{\pi'}$
must be $|A_{\pi'}| = x_{\pi'}$.
Assuming that $\fA$ is globally homogeneous, we may fixate on the fact
that the shared profile $\rho_\pi$ of $\pi$ has exactly $|A_{\pi'}|$ witnesses for fluted 2-types with the endpoint $\pi'$, or, more formally,
$\sum_{\tau \in \FTP^\sigma_2}^{\tau \restriction_{[2,2]} = \pi'} \rho(\tau) = |A_{\pi'}|$.
Thus, the assignment $y_{\pi, \tau} := \rho_\pi(\tau)$ satisfies the following set of inequations:
\begin{equation*}\tag{$\Psi_2$}\label{eq_FLPC2_Psi2}
    \Big\{
        x_\pi \neq 0 \to
        \sum_{\tau \in \FTP^\sigma_2}^{\tau \restriction_{[2,2]} = \pi'} y_{\pi, \tau} = x_{\pi'} \; \Big\vert \;
        \pi, \pi' \in \FTP^\sigma_1
    \Big\}.
\end{equation*}
Of course, under the supposition that $\pi \models \alpha_r$ for some $r \in R$ and $\pi \in \FTP^\sigma_1$ such that $|A_\pi| \geq 1$, we have that
$\rho_\pi \models \exists_{[n_r^{+p_r}]} \gamma_r$. Clearly, $k := \sum_{\tau \in \FTP^\sigma_2}^{\tau \models \gamma_r} \rho_\pi(\tau)$
must be a member of the linear set $n_r^{+p_r}$. Thus, there is a period counter $i_{\pi, r} \in \N$ such that
$k = n_r {+} i_{\pi, r} p_r$.
Turning the equation around we get that
$i_{\pi, r} = (k {-} n_r) / p_r$. Recalling that $\rho_\pi(\tau) = y_{\pi, \tau}$ for all $\tau \in \FTP^\sigma_2$, we have that the following is satisfied by our assignments:
\begin{equation}
    \Big\{
        x_\pi \neq 0 \to
        \sum_{\tau \in \FTP^\sigma_2}^{\tau \models \gamma_r} y_{ \pi, \tau } = n_r {+} i_{\pi, r} p_r \; \Big\vert \;
        r \in R, \pi \in \FTP^\sigma_1 \text{ s.t. } \pi \models \alpha_r
    \Big\} .\tag{$\Psi_3$}\label{eq_FLPC2_Psi3}
\end{equation}
On the other hand, supposing $\pi \models \beta_t$ for some $t \in T$ and with $\pi \in \FTP^\sigma_1$ such that $|A_\pi| \geq 1$,
we have that $\rho_\pi \not\models \exists_{[n_t^{+p_t}]} \delta_r$, thus leaving
$k := \sum_{\tau \in \FTP^\sigma_2}^{\tau \models \gamma_r} \rho_\pi(\tau)$ outside of the linear set $n_t^{+p_t}$.
Notice that this happens when the following conditions are met:
\begin{enumerate}
    \item $k < n_t$; or \label{item:non_incl_1}
    \item $p_t \neq 0$ and $k > m$ for each $m \in n_t^{+p_t}$ (which only happens when $k = \aleph_0$); or \label{item:non_incl_2}
    \item $p_t = 0$ and $k > n_t$; or \label{item:non_incl_3}
    \item for some $m \in n_t^{+p_t}$ we have $m < k < m + p_t$.\label{item:non_incl_4}
\end{enumerate}
Note that the listed conditions are exhaustive.
We translate the requirements \ref{item:non_incl_1}--\ref{item:non_incl_4}
into four functions $\Theta_1, \dots, \Theta_4$ which map fluted 1-types paired with indices in $T$
to linear equations. The functions are then defined as follows:
\begin{itemize}
    \item $\Theta_1(\pi, t) := \sum_{\tau \in \FTP^\sigma_2}^{\tau \models \delta_t} y_{ \pi, \tau } < n_t$,
    \item $\Theta_2(\pi, t) := \sum_{\tau \in \FTP^\sigma_2}^{\tau \models \delta_t} y_{ \pi, \tau } = \aleph_0$,
    \item $\Theta_3(\pi, t) := (p_t = 0) \wedge (n_t < \sum_{\tau \in \FTP^\sigma_2}^{\tau \models \delta_t} y_{ \pi, \tau })$,
    \item $\Theta_4(\pi, t) := n_t + j_{\pi, t} p_t < \sum_{\tau \in \FTP^\sigma_2}^{\tau \models \delta_t} y_{ \pi, \tau } < n_t + (j_{\pi, t} + 1) p_t$.
\end{itemize}
Since $k$ adheres to at least one of the four conditions, we write the following clauses for eligible fluted $1$-types:
\begin{equation}
    \Big\{
        x_\pi \neq 0 \to \bigvee_{i \in [1, 4]} \Theta_i(\pi, t)
        \; \Big\vert \; t \in T \text{ and } \pi \in \FTP^\sigma_1 \text{ such that } \pi \models \beta_t
    \Big\}.\tag{$\Psi_4$}\label{eq_FLPC2_Psi4}
\end{equation}
To verify that $j_{\pi, \tau} = \lfloor( k - n_t ) / p_t \rfloor$ is indeed a satisfying assignment for $\Psi_4$,
we need only consider case \ref{item:non_incl_4}. For this assume that $m < k < m + p_t$ for some $m \in n_t^{+p_t}$.
We can thus write $m = n_t + j p_t$.
Since $j_{\pi, \tau} = \lfloor( k - n_t ) / p_t \rfloor$, we conclude that $j_{\pi, \tau} = j$
thus satisfying $\Theta_4(\pi, t)$.

Reflecting on the semantics of the equality predicate, we see that for any given $\pi \in \FTP^\sigma_1$ in any globally homogeneous model $\fA \models \varphi$
there is exactly one $\tau \in \FTP^\sigma_2$ such that $\tau$ features the non-negated equality predicate and
$\rho_\pi(\tau) \neq 0$. More precisely, $\rho_\pi(\tau) = 1$ and the endpoint of $\tau$ is $\pi$.
We have that our assignment respects this condition and thus satisfies the following inequations:
\begin{equation}
    \Big\{
        y_{ \pi, \tau } = 0 \; \Big\vert \;
        \pi \in \FTP^\sigma_1, \; \tau \in \FTP^\sigma_2 \text{ s.t. } {=} \in \tau, \tau \restriction_{[2,2]} \neq \pi
    \Big\} \cup
    \Big\{
        \hspace{-2mm} \sum_{\tau \in \FTP^\sigma_2}^{{=} \in \tau} \hspace{-2mm} y_{ \pi, \tau } = 1
        \; \Big\vert \;
        \pi \in \FTP^\sigma_1
    \Big\}. \tag{$\Psi_5$}\label{eq_FLPC2_Psi5}
\end{equation}
Finally, we forbid the periodic counters
$(i_{\pi, r})_{\pi \in \FTP^\sigma_1}^{r \in R}$ and $(j_{\pi, t})_{\pi \in \FTP^\sigma_1}^{t \in R}$
from taking the value~$\aleph_0$:
\begin{equation}
    \Big\{
        i_{ \pi, r }, j_{ \pi, t } < \aleph_0 \; \Big\vert \;
        \pi \in \FTP^\sigma_1 \text{ and }  r \in R , t \in T
    \Big\}. \tag{$\Psi_6$}\label{eq_FLPC2_Psi6}
\end{equation}  

Putting everything together, we have engineered a system of equations $\Psi = \Psi_1 \cup \cdots \cup \Psi_6$
that is satisfied by extracting relevant cardinalities (see (\ref{FLPC2_assignments})) from homogeneous models of $\varphi$.
\begin{lemma}\label{flpc2_reduct_dir_1}
    Suppose $\fA \models \varphi$ is a globally homogeneous model. Then, $\Psi$ is satisfiable.
\end{lemma}

We now move on to the converse direction:

\begin{lemma}\label{flpc2_reduct_dir_2}
    Suppose $\Psi$ is satisfiable. Then, there is a globally homogeneous model $\fA \models \varphi$.
\end{lemma}
\begin{proof}
    Suppose that $\Psi$ has a satisfying assignment.
    To avoid notational clutter, we identify the solution vectors of variables in $\Psi$ as themselves.
    We will build a globally homogeneous model $\fA$ over the domain
    $
        A = \bigcup_{\pi \in \FTP^\sigma_1} A'_\pi, \text{ where } A'_\pi = \{ a_{\pi, i} \mid i \in [1, x_{\pi}] \}.
    $
    Intuitively, we wish that elements of $A'_\pi$ realise the fluted 1-type $\pi$.
    We thus assign $\ftp^{\fA}_1(a) := \pi$ for all $a \in A'_\pi$.
    Recalling that $A_\pi$ is the set of all elements that realise the 1-type $\pi$ in $\fA$,
    we have that $A_\pi = A'_\pi$ is of cardinality $x_\pi$ as required by (\ref{FLPC2_assignments_2}).
    By \ref{eq_FLPC2_Psi1}, the domain is non-empty.

    Picking any $\pi \in \FTP^\sigma_1$ we now move on to the assignment of fluted 2-types.
    Take any $\pi' \in \FTP^\sigma_1$ and let $S = \{ \tau \in \FTP^\sigma_2 \mid \tau \restriction_{[2,2]} = \pi' \}$,
    i.e. $S$ is the set of all fluted $2$-types containing the fluted 1-type $\pi'$ as an endpoint.
    By \ref{eq_FLPC2_Psi2}, we have that $\sum_{\tau \in S} y_{\pi, \tau} = x_{\pi'}$,
    thus $\sum_{\tau \in S} y_{\pi, \tau} = |A_{\pi'}|$.
    Now, pick some element $a \in A_\pi$.
    In case $\pi \neq \pi'$, equation \ref{eq_FLPC2_Psi5} prohibits fluted 2-types that feature the (non-negated) equality literal.
    We set fluted 2-types between $a$ and elements of $A_{\pi'}$ in any way that results in
    $|\{ b \in A_{\pi'} \mid \ftp^\fA_2(ab) = \tau \}| = y_{\pi, \tau}$ for each $\tau \in S$
    (n.b. the exact configuration of fluted 2-types between $a$ and elements of $A_{\pi'}$
    is irrelevant as the fluted 2-type of $ba$ for any $b \in A_{\pi'}$ is not set in this process).
    In case $\pi = \pi'$ notice that by \ref{eq_FLPC2_Psi5} there is exactly one $\tau^= \in S$ such that
    \begin{enumerate*}[label=(\roman*)]
        \item ${=} \in \tau^=$,
        \item $y_{\pi, \tau^=} \geq 0$, and with
        \item $\tau^= \restriction_{[2,2]} = \pi$.
    \end{enumerate*}
    By \ref{eq_FLPC2_Psi5} again, we have that $y_{\pi, \tau^=} = 1$.
    We therefore set the fluted 2-types between $a$ and $A_{\pi} \setminus \{ a \}$
    for each $\tau \in S \setminus \{ \tau^= \}$
    as in the case before and, additionally, specify that $\ftp^\fA_2(aa) := \tau^=$.

    By repeating the fluted 2-type assignment for each element $a \in A$ and fluted 1-type $\pi' \in \FTP^\sigma_1$
    we are guaranteed that elements in $A_\pi$ (where $\pi = \ftp^\fA(a)$) realise the fluted
    1-profile~$\rho_\pi := \{ \tau \mapsto y_{\pi, \tau} \mid \tau \in \FTP^\sigma_2 \}$ as required by (\ref{FLPC2_assignments_2}).
    The resulting structure is clearly globally homogeneous.

    We now claim that the resulting structure is a model of $\varphi$. Indeed, take any $a \in A$ with $\pi = \ftp^\fA_1(a)$ and
    suppose $\pi \models \alpha_r$ for some $r \in R$. Let $S = \{ \tau \in \FTP^\sigma_2 \mid \tau \models \gamma_r\}$.
    By equations \ref{eq_FLPC2_Psi3} and \ref{eq_FLPC2_Psi6}, the sum $\sum_{\tau \in S} y_{\pi, \tau}$ is a member of the linear set $n_r^{+p_r}$.
    Since the element $a$ is of fluted 1-type $\pi$, we have that it realises the profile $\rho_\pi$ in $\fA$.
    By our construction, $\rho_\pi(\tau) = y_{\pi, \tau}$ for each $\tau \in \FTP^\sigma_2$.
    Thus,
    $\rho_\pi \models \exists_{[n_r^{+p_r}]} \gamma_r$ which is equivalent to
    $\fA, a \models \exists_{[n_r^{+p_r}]} \gamma_r$ as required.

    On the other hand, suppose $\pi \models \beta_t$ for some $t \in T$. We claim that $\fA, a \not\models \exists_{[n_t^{+p_t}]} \delta_t$.
    To see this, let $S$ be the set $\{ \tau \in \FTP^\sigma_2 \mid \tau \models \delta_t \}$.
    Writing $k = \sum_{\tau \in S} y_{\pi, \tau}$
    we take note of equations \ref{eq_FLPC2_Psi4} and \ref{eq_FLPC2_Psi6}, and conclude that one of the following conditions must be true:
    \begin{enumerate}
        \item $k$ is smaller than the minimal element of $n_t^{+p_t}$; or
        \item $n_t^{+p_t} \subseteq \mathbb{N}$ and $k = \aleph_0$; or
        \item $p_t = 0$ and $k > n_t$; or
        \item $k$ is in between two consecutive elements of $n_t^{+p_t}$.
    \end{enumerate}
    Whichever case it may be, we have that $k \not\in n_t^{+p_t}$.
    Again, recalling that $\rho_\pi(\tau) = y_{\pi, \tau}$ for each $\tau \in \FTP^\sigma_2$,
    we conclude that $\rho_\pi \not\models \exists_{[n_t^{+p_t}]} \delta_t$ which is equivalent to saying $\fA, a \not\models \exists_{[n_t^{+p_t}]} \delta_t$.
\end{proof}

Given an $\FLPC^2$-sentence $\varphi$ we present a decision procedure for the (finite) satisfiability problem.
Compute a normal-form formula $\psi$ from $\varphi$ and
write the linear Diophantine equations $\Psi$.
Now, guess a solution vector $\bar{z}$ which can be done in non-deterministic polynomial time as a function of $||\Psi||$.
If $\bar{z}$ is indeed a solution for $\Psi$, accept,
otherwise, reject. In the case of the finite satisfiability problem, prohibit $\aleph_0$ from being a solution in $\Psi$.
Correctness of the procedure follows from the fact that,
by Lemma~\ref{cor:glob_homo}, $\psi$ if satisfiable then it is satisfiable in a globally homogeneous model.
Thus, by Lemma~\ref{flpc2_reduct_dir_1}, if $\psi$ is satisfiable, then $\Psi$ has a solution.
On the other hand, by Lemma~\ref{flpc2_reduct_dir_2}, if $\Psi$ has a solution, then $\psi$ is satisfiable.

Noting that the satisfiability problem for $\FL^2$ is \textsc{NExpTime}-hard \cite{phst19},
and that $||\Psi||$ is bounded by a polynomial function on the number of different fluted 1- and 2-types (of which there are $2^{||\varphi||}$ many),
we conclude the following:
\begin{theorem} \label{th:2_var_NExp}
    The (finite) satisfiability problem of $\FLPC^2$ is \textsc{NExpTime}-complete.
\end{theorem}
\section{More Than Two Variables}\label{sec:more}

We now generalise our results on homogeneity and decidability of satisfiability for higher-arity formulas of $\FLPC$.
Thus, throughout this section, we will be working in the $(\ell{+}1)$-variable sub-fragment of $\FLPC$, where $\ell \geq 2$ is fixed.

Firstly, we lift our homogeneity conditions to the multivariable setting.
Suppose $\fA$ is a $\sigma$-structure and take any $(\ell{-}1)$-tuple $\bar{b}$ from $A$ and $\pi \in \FTP^\sigma_\ell$.
Let $A_{\pi \leftarrow \bar{b}}$ be the set $\{ a \in A \mid \ftp^\fA(a\bar{b}) = \pi \}$.
We say that $\pi$ is \textit{$\bar{b}$-homogeneous} in $\fA$ if for each $a, a' \in A_{\pi \leftarrow \bar{b}}$
and all $c \in A$ we have that $\ftp^\fA(a\bar{b}c) = \ftp^\fA(a'\bar{b}c)$.
That is to say, $\bar{b}c$ absorbs the same fluted $(\ell{+}1)$-type from each $\ell$-tuple $a\bar{b}$ that realises the fluted $\ell$-type $\pi$.
If each $\pi \in \FTP^\sigma_\ell$ is $\bar{b}$-homogeneous in $\fA$, then we say that the $(\ell{-}1)$-tuple $\bar{b}$ is \textit{homogeneous} in $\fA$.
Finally, if each $(\ell{-}1)$-tuple $\bar{b}$ is homogeneous in $\fA$, then we say that $\fA$ is \textit{locally $\ell$-homogeneous}.

When considering satisfiable normal-form $\FLPC^{\ell{+}1}$-sentences we can, without loss of generality, confine ourselves to
locally $\ell$-homogeneous structures. To see this fix some normal-form $\FLPC^{\ell{+}1}$-sentence $\varphi$
 and take $\bar{b} \in A^{\ell{-}1}$ and $a, a' \in A_{\pi \leftarrow \bar{b}}$.
Proceeding similarly as before Lemma~\ref{cor:glob_homo} we see that the fluted $(\ell{+}1)$-type of $a'\bar{b}c$
does not impact the satisfaction of quantifier-free $\FLPC^{\ell{+}1}$-formulas by $c\bar{b}^{-1}a'$.
Thus, redefining the fluted $(\ell{+}1)$-types emitted by $a'\bar{b}$ will not alter the satisfaction of quantifier-free $\FLPC^{\ell{+}1}$-formulas by other tuples.
Notice again that it is impossible to prohibit $\ftp^\fA(a\bar{b}c) \neq \ftp^\fA(a'\bar{b}c)$ by a normal-form formula.
Thus, by setting $\ftp^\fA(a'\bar{b}c) := \ftp^\fA(a\bar{b}c)$ for each $c \in A$ and repeating the procedure for each
$a' \in A_{\pi \leftarrow \bar{b}} \setminus \{ a \}$, we will have that $\pi$ is \textit{$\bar{b}$-homogeneous} in $\fA$.
Clearly $\fA \models \varphi$ as the tuples rewired by this procedure (i.e. $a'\bar{b}$ with $a' \in A_{\pi \leftarrow \bar{b}} \setminus \{ a \}$)
now emit $\tau \in \FTP^\sigma_{\ell{+}1}$ to a given witness if and only if $a\bar{b}$ does.
The rewired tuples thus satisfy the same exact $\FLPC^{\ell{+}1}$-formulas with at most 1 quantifier as $a\bar{b}$ does.
Repeating the procedure for all $\pi \in \FTP^\sigma_\ell$ and $\bar{b} \in A^{\ell{-}1}$
we will have the~following:
\begin{lemma} \label{cor:loc_homo}
    Suppose $\varphi$ is a satisfiable normal-form $\FLPC^{\ell{+}1}$-sentence. Then, $\varphi$ is satisfiable in a locally $\ell$-homogeneous model.
\end{lemma}

Using local $\ell$-homogeneity coupled with variable reduction techniques prevalent in studies of fluted logics (see \cite{phst19}),
we will establish a decidability result for the (finite) satisfiability problem of $\FLPC^{\ell{+}1}$.
More specifically, fixing a normal-form $\FLPC^{\ell{+}1}$-sentence $\varphi$ we will compute a normal-form $\FLPC^\ell$-sentence $\psi$
that is satisfiable in structures holding just enough information to build locally $\ell$-homogeneous models for $\varphi$.
To aid motivation, we fix
$\fA$ to be any locally $\ell$-homogeneous model of $\varphi$.
We shall construct $\psi$ whilst also expanding $\fA$ into $\fA' \models \psi$.
Note that the construction depends exclusively on the syntactic properties of $\varphi$.

First, set $\fA' := \fA$.
Take $(q_{\pi})_{\pi \in \FTP^\sigma_\ell}$ to be a sequence of fresh $(\ell{-}1)$-ary predicate symbols.
In $\fA'$ we decorate $(\ell{-}1)$-tuples $\bar{b}$ over $A$ with $q_\pi$ just in case we have that $A_{\pi \leftarrow \bar{b}} \neq \emptyset$.
That is to say, $q_\pi^{\fA'}$ remembers which $(\ell{-}1)$-tuples can be extended (by appending an element to the left)
to realise the fluted $\ell$-type $\pi$.
It is clear that $\fA'$ models the following:
\begin{equation}
    \bigwedge_{\pi \in \FTP^\sigma_\ell} \forall^\ell \Big( \pi \to q_\pi \Big). \tag{$\psi_1$} \label{eq:red_sent1}
\end{equation}

Proceeding similarly, let $(s_{\pi, \tau})^{\tau \in \FTP^\sigma_{\ell{+}1}}_{\pi \in \FTP^\sigma_\ell}$ be a sequence of new $\ell$-ary
predicates. Intuitively, we will have $\bar{b}c \in s_{\pi, \tau}^{\fA'}$ if in $\fA$ it is the case that $\bar{b}c$ absorbs
the fluted $(\ell{+}1)$-type $\tau$ emitted from $a\bar{b}$ for some $a \in A_{\pi \leftarrow \bar{b}}$.
Notice that, by local $\ell$-homogeneity, if $\bar{b}c$ absorbs $\tau$ from some $a\bar{b}$ with $a \in A_{\pi \leftarrow \bar{b}}$,
then it absorbs $\tau$ from $a'\bar{b}$ for all $a' \in A_{\pi \leftarrow \bar{b}}$.
Thus, by our construction, $s_{\pi, \tau}$ is the unique predicate amongst $(s_{\pi, \tau'})^{\tau' \in \FTP^\sigma_{\ell{+}1}}$
satisfied by $\bar{b}c$ in $\fA'$.
Clearly, $\fA'$ models:
\begin{equation}    
    \begin{split}
        \bigwedge_{\pi \in \FTP^\sigma_\ell} \bigwedge_{\tau \in \FTP^\sigma_{\ell{+}1}} &\forall^{\ell\phantom{-1}} \Big( s_{\pi, \tau} \to \tau \restriction_{[2,\ell{+}1]} \Big) \wedge \mbox{} \\
        \bigwedge_{\pi \in \FTP^\sigma_\ell} &\forall^{\ell{-}1} \Big( q_\pi \to \forall \big( \bigvee_{\tau \in \FTP^\sigma_{\ell{+}1}} s_{\pi, \tau} \wedge
        \bigwedge_{\tau, \tau' \in \FTP^\sigma_{\ell{+}1}}^{\tau \neq \tau'} (\neg s_{\pi, \tau} \vee \neg s_{\pi, \tau'})
        \big) \Big).
    \end{split}  \tag{$\psi_2$}\label{eq:red_sent2}
\end{equation}

Again taking $\bar{b} \in A^{\ell{-}1}$ and any $\pi \in \FTP^\sigma_\ell$ suppose $\pi \models \alpha_r$ for some $r \in R$.
In case $A_{\pi \leftarrow \bar{b}}$ is non-empty (thus guaranteeing $\bar{b} \in q_\pi^{\fA'}$), we pick any $a \in A_{\pi \leftarrow \bar{b}}$ and write
$S = \{ c \in A \mid \fA, a\bar{b}c \models \gamma_r \}$.
Since $\pi$ is $\bar{b}$-homogeneous in $\fA$, the exact element in $A_{\pi \leftarrow \bar{b}}$ we pick has no effect on $S$.
By our construction, $S$ is then exactly the set of element $c \in A$ such that
$\fA', \bar{b}c \models \bigvee_{\tau \in \FTP^\sigma_{\ell{+}1}}^{\tau \models \gamma_r} s_{\pi, \tau}$.
Since $|S| \in n_r^{{+}p_r}$ it is then immediate that $\fA'$ models the following: 
\begin{equation}
    \bigwedge_{r \in R} \bigwedge_{\pi \in \FTP^\sigma_{\ell}}^{\pi \models \alpha_r} \forall^{\ell{-}1}
    \Big(
        q_\pi \to \exists_{[n_r^{+p_r}]} \bigvee_{\tau \in \FTP^\sigma_{\ell{+}1}}^{\tau \models \gamma_r} s_{\pi, \tau}
    \Big). \tag{$\psi_3$} \label{eq:red_sent3}
\end{equation}
Similar observations follow whenever $\pi \models \beta_t$ for some $t \in T$.
This time, however, the cardinality of $S = \{ c \in A \mid \fA, a\bar{b}c \models \delta_r \}$
must be outside the set $n_r^{+p_r}$.
Clearly, $\fA'$ models:
\begin{equation}
    \bigwedge_{t \in T} \bigwedge_{\pi \in \FTP^\sigma_{\ell}}^{\pi \models \beta_t} \forall^{\ell{-}1}
    \Big(
        q_\pi \to \neg \exists_{[n_t^{+p_t}]} \bigvee_{\tau \in \FTP^\sigma_{\ell{+}1}}^{\tau \models \delta_t} s_{\pi, \tau}
    \Big). \tag{$\psi_4$} \label{eq:red_sent4}
\end{equation}

Writing $\psi := \psi_1 \wedge \cdots \wedge \psi_4$ we have shown the following:
\begin{lemma}\label{FLPC_var_reduct_1}
    Suppose $\fA \models \varphi$ is a locally $\ell$-homogeneous model. Then, $\fA$ can be extended to a model $\fA'$ of
    $\psi$.
\end{lemma}

\begin{lemma}\label{FLPC_var_reduct_2}
    Suppose $\fA' \models \psi$. Then, we can construct a locally $\ell$-homogeneous model $\fA^+$ of $\varphi$
    over the same domain.
\end{lemma}
\begin{proof}
    Supposing $\psi$ is satisfiable we take any model $\fA'$.
    Now, let $\fA^-$ be the model $\fA'$ but with the predicates in $(q_\pi)_{\pi \in \FTP^\sigma_{\ell}}$ and $(s_{\pi, \tau})_{\pi \in \FTP^\sigma_{\ell}}^{\tau \in \FTP^\sigma_{\ell{+}1}}$
    removed from the signature.
    We proceed by expanding $\fA^-$ into a locally $\ell$-homogeneous model $\fA^+$ of the original sentence
    $\varphi$.
    
    Fix $\bar{b} \in A^{\ell{-}1}$ and take some $a \in A$.
    Supposing that $\ftp^{\fA^-}_{\ell}(a\bar{b}) = \pi$,
    by \ref{eq:red_sent1} we have that $\fA', \bar{b} \models q_\pi$.
    Taking any $c \in A$ we observe that the conjuncts of \ref{eq:red_sent2} enforce the following:
    \begin{itemize}
        \item if $\fA', \bar{b}c \models s_{\pi, \tau}$ for some $\tau \in \FTP^\sigma_{\ell{+}1}$,
            then $\bar{b}c$ can absorb the fluted $(\ell{+}1)$-type $\tau$,
        \item $\bar{b}c$ satisfies at least one of the predicates $(s_{\pi, \tau})_{\tau \in \FTP^\sigma_{\ell{+}1}}$, and
        \item $\bar{b}c$ satisfies at most one of the predicates $(s_{\pi, \tau})_{\tau \in \FTP^\sigma_{\ell{+}1}}$.
    \end{itemize}
    We can then safely set $\ftp^{\fA^+}_{\ell{+}1}(a\bar{b}c) := \tau$ for each $c \in A$,
    where $\tau$ is taken from the subscript of the unique $s_{\pi, \tau} \in (s_{\pi, \tau})_{\tau \in \FTP^\sigma_{\ell{+}1}}$
    that $\bar{b}c$ satisfies in $\fA'$.
    By repeating the above procedure for all $a \in A$ and tuples
    $\bar{b} \in A^{\ell{-}1}$ we will obtain the desired structure $\fA^+$.

    To verify that $\fA^+$ is a model of $\varphi$ we first claim that 
    $
        \fA^+ \models \bigwedge_{r \in R} \forall^\ell ( \alpha_r \to \exists_{[n_r^{+p_r}]} \gamma_r ).
    $
    For this purpose, fix some $r \in R$ and $a\bar{b} \in A^\ell$, and suppose
    $\pi \models \alpha_r$, where $\ftp^{\fA^+}_\ell(a\bar{b}) = \pi$.
    Recall that $\bar{b} \in q_\pi^{\fA'}$ by \ref{eq:red_sent1}.
    Then, \ref{eq:red_sent3} gives us
    $\fA', \bar{b} \models \exists_{[n_r^{+p_r}]} \bigvee_{\tau \in \FTP^\sigma_{\ell{+}1}}^{\tau \models \gamma_r} s_{\pi, \tau}$.
    Taking any $c \in A$ we have, by our construction, that $\fA^+, a\bar{b}c \models \tau$ if and only if $\fA', \bar{b}c \models s_{\pi, \tau}$.
    Thus, $\fA^+, a\bar{b} \models \exists_{[n_r^{+p_r}]} \bigvee_{\tau \in \FTP^\sigma_{\ell{+}1}}^{\tau \models \gamma_r} \tau$,
    which is equivalent to saying $\fA^+, a\bar{b} \models \exists_{[n_r^{+p_r}]} \gamma_r$.
    Repeating the argument for each $r \in R$ and $a\bar{b} \in A^\ell$ will yield the required result.

    To show
    $
        \fA^+ {\models} \bigwedge_{t \in T} \forall^\ell ( \beta_t {\to} \neg \exists_{[n_t^{+p_t}]} \delta_t )   
    $
    we proceed analogously with \ref{eq:red_sent4} in place of~\ref{eq:red_sent3}.
\end{proof}

Let us take stock of the previous three lemmas.
Take $\varphi$ to be an $\FLPC^{\ell{+}1}$-sentence. Without loss of generality,
assume that it is in normal-form (Lemma~\ref{lma:nmf}).
By Lemma~\ref{cor:loc_homo}, $\varphi$ is satisfiable if it is satisfiable in a locally $\ell$-homogeneous model.
By computing the formula $\psi$ we have, by Lemmas~\ref{FLPC_var_reduct_1}~and~\ref{FLPC_var_reduct_2},
that $\psi$ is (finitely) satisfiable if and only if $\varphi$ is.
Noting that the (finite) satisfiability problem for $\FLPC^2$ is in \textsc{NExpTime} (Theorem~\ref{th:2_var_NExp}) and that
$\psi$ can be constructed from $\varphi$ in polynomial time in regards to the number of different fluted $\ell$- and $(\ell{+}1)$-types
(of which there are $2^{O(||\varphi||)}$),
we conclude the following:
%
\begin{theorem} \label{th:more_vars_time}
    The (finite) satisfiability problem for $\FLPC^{\ell{+}1}$ is in $\ell$-\textsc{NExpTime}.
\end{theorem}

Noting that the (finite) satisfiability problem for $\FL^{\ell{+}1}$ is $\lfloor (\ell{+}1) / 2 \rfloor$-\textsc{NExpTime}-hard \cite{phst19},
we see that no elementary function can encapsulate the complexity of (finite) satisfiability for $\FLPC$.
We thus conclude our section having reached our initial goal:
\begin{theorem}
    The (finite) satisfiability problem for $\FLPC$ is \textsc{Tower}-complete.
\end{theorem}

\section{Counting With Reversed Relations}\label{sec_undec}

In this section we will show that relaxing the syntactic restrictions of the fluted fragment with counting yields undecidability of satisfiability.
We define the language $\FL_{\text{rev}}$ to be $\FL$ but with the addition of atoms with reversed variable sequences.
More formally, if $r(x_k, \dots, x_\ell)$ is an $\FL$-atom, then $r(x_k, \dots, x_\ell)$ and $r(x_\ell, \dots, x_k)$ are 
$\FL_{\text{rev}}$-atoms. The language $\FLC_{\text{rev}}$ ($\FLPC_{\text{rev}}$) is then the obvious extension of $\FL_{\text{rev}}$ with (periodic)
counting quantifiers. Clearly, the languages $\FL_{\text{rev}}$, $\FLC_{\text{rev}}$, and $\FLPC_{\text{rev}}$ are subfragments of the adjacent fragment
with the appropriate counting extensions.
We will use counting quantifiers $\exists_{[=1]}$ and $\exists_{[\leq 1]}$ with the meanings ``there is exactly one element s.t. ...'' and 
``there is at most one element s.t. ...'' along side periodic counting qunatifiers.
For simplicity, we do away with variable-free notation and use variable sequences of $x, y, z$ and $z, y, x$ in place of $x_1, x_2, x_3$.

We proceed by reducing \textit{Hilbert's 10th problem} to the finite satisfiability problem of $\FLC^3_{\text{rev}}$.
    Let $\mathcal{E}$ be a system of Diophantine equations. We assume that each equation $e \in \mathcal{E}$
    is of one of the following (simple) forms:
    \begin{enumerate*}[label=(\roman*)]
        \item $u = 1$,\label{hilb_1}
        \item $u + v = w$, or\label{hilb_2}
        \item $u \cdot v = w$,\label{hilb_3}
    \end{enumerate*}
    where $u, v, w$ are mutually disjoint variables.
    Clearly, no loss of generality occurs as any (non-simple) Diophantine equation can be rewritten into the simpler form by introducing new
    variables. For each $e \in \mathcal{E}$ we will define a formula $\varphi_e$ depending on the form that $e$ takes.
    Then, $\varphi := \bigwedge_{e \in \mathcal{E}} \varphi_e \wedge \psi$ will be the advertised formula that is finitely satisfiable if and only if $\mathcal{E}$ has
    a solution over $\mathbb{N}$. We specify that the signature of $\varphi$ includes
    \begin{enumerate*}[label=(\arabic*)]
        \item unary predicates $A_u$ for each variable of $u$ in $\mathcal{E}$,
        \item binary predicates $R_e$ for each $e \in \mathcal{E}$ of the form \ref{hilb_2}, and
        \item ternary predicates $P_e$ for each $e \in \mathcal{E}$ of the form \ref{hilb_3}.
    \end{enumerate*}
    We will not assume that the equality predicate is available.
    In the sequel we will argue that if $\fA \models \varphi$, then
    $\mathcal{E}$ has a solution with $u \mapsto |A_u^\fA|$ for each variable $u$. (And, of course, the converse as well). For technical reasons we wish for the sets
    $A_u^\fA$ and $A_v^\fA$ with $u \neq v$ to be disjoint. Denoting $\text{vars}(\mathcal{E})$ for the set of variables in $\mathcal{E}$, we first define
    $
       \psi :=         \bigwedge_{u, v \in \text{vars}(\mathcal{E})}^{u \neq v} \forall x  \big( \lnot A_u(x) \vee \lnot A_v(x) \big),
    $
    which clearly has the required effect.
    We proceed by taking $e \in \mathcal{E}$ in turn.

    Suppose first that $e$ is of the form \ref{hilb_1} $u = 1$.
    We ensure that every model $\fA$ of $\varphi$ will have $|A_u^\fA| = 1$ by defining $\varphi_e$ to be
    $
    \exists_{[=1]} x \; A_u(x).
    $

    Now, supposing that $e$ is of the form \ref{hilb_2} $u + v = w$,
    we define $\varphi_e$ with the intent that models $\fA$ of $\varphi$ will have $R_e^\fA$ being a bijection between $A_u^\fA \cup A_v^\fA$ and $A_w^\fA$ (i.e. $|A_u^\fA| + |A_v^\fA| = |A_w^\fA|$):
    \begin{align*}
        \forall x &\Big( \big( A_u(x) \vee A_v(x) \big) \to \exists_{[=1]} y \big( A_w(y) \wedge R_e(xy) \big) \Big) \wedge \mbox{}\\
        \forall y  &\Big( A_w(y) \to \exists_{[=1]} x  \big( ( A_u(x) \vee A_v(x) ) \wedge R_e(xy) \big) \Big).
    \end{align*}

    Lastly, if $e$ is of the form \ref{hilb_3} $u \cdot v = w$, then $\varphi_e$ (defined just below) will guarantee that
    $\fA \models \varphi$ forces $P_e^\fA$ to be a bijection between $A_u^\fA \times A_v^\fA$ and $A_w^\fA$ (i.e. $|A_u^\fA| \cdot |A_v^\fA| = |A_w^\fA|$):
    \begin{equation*}
        \begin{split}
            \forall x &\Big( A_u(x) \to \forall y \big( A_v(y) \to
            \exists_{[=1]} z  ( A_w(z) \wedge P_e(xyz) ) \big) \Big) \wedge \mbox{} \\
            \forall z  &\Big( A_w(z) \to \exists_{[=1]} y  \big( A_v(y) \wedge 
            \exists x  ( A_u(x) \wedge P_e(xyz) ) \big) \Big)  \wedge \mbox{} \\
            \forall z  & \Big( A_w(z) \to \exists y  \big( A_v(y) \wedge
            \exists_{[=1]} x  ( A_u(x) \wedge P_e(xyz) ) \big) \Big).
        \end{split}
    \end{equation*}

    It is then straightforward to show (see Appendix~\ref{app:undec_missing}) that $\varphi$ is finitely satisfiable iff $\mathcal{E}$
    has a solution over $\mathbb{N}$.
    Noting that the problem of finding solutions to Diophantine equations over $\mathbb{N}$ is $\Sigma^0_1$-complete,
    and that $\varphi \wedge \exists_{[0^{+1}]} x \; \top$ is an $\FLPC^3_{\text{rev}}$-sentence that is satisfiable if and only if
    $\varphi$ is finitely satisfiable,
    we conclude the following:
    \begin{theorem} \label{th_FL_rev_fin_undec}
        The finite satisfiability problem for $\FLC^3_{\text{rev}}$ is $\Sigma^0_1$-hard.
        If periodic counting is permitted, then so is the general satisfiability problem.
    \end{theorem}

We note that one can use the same type of argument as above
when reducing from the problem of solving Diophantine equations $\mathcal{E}$ over $\mathbb{N}^* = \{1, 2, \dots\} \cup \{ \aleph_0 \}$
to the general satisfiability problem of $\FLC+\cdot^{-1}$.
Such an approach, however, is not fruitful for determining undecidability as the problem of finding solutions to $\mathcal{E}$ over $\mathbb{N}^*$ is in \textsc{NPTime} \cite{dioph_infinity_decidable}.
Thus, to show undecidability of general satisfiability, we resort to the tiling problem.
Take $\Phi = \langle \mathcal{T}, \mathcal{H}, \mathcal{V} \rangle$ with $\mathcal{H}, \mathcal{V} \subseteq \mathcal{T} \times \mathcal{T}$.
We will produce and $\FLC^4_{\text{rev}}$-sentence $\varphi$ that
is satisfiable if and only if $\Phi$ tiles the infinite $\NN$ plane in accordance to the horizontal ($\mathcal{H}$) and vertical~($\mathcal{V}$) constraints.
Additionally, we will argue that one can append additional $\FLPC^2_{\text{rev}}$ conjuncts to
$\varphi$ and thus obtain a reduction from the tiling problem with a designated tile recurring infinitely often on the first column (see \cite{HAREL198551} for details about the problem).
All-in-all, such reductions will guarantee that $\FLC^4_{\text{rev}}$ is $\Pi^0_1$-hard, 
whilst making $\FLPC^4_{\text{rev}}$ hard for $\Sigma^1_1$.

Take $G$ to be unary and $H$, $V$ to be binary predicate symbols.
We define the \textit{canonical \NNs-grid} to be a $\{ G, H, V \}$-structure $\fG$ over the domain $\NN$ with the following extensions:
\begin{itemize}
    \item $G^\fA := \NN$,
    \item $H^\fA := \{ \langle (i, j), (i, j{+}1) \rangle \mid i, j \in \N \}$, and
    \item $V^\fA := \{ \langle (i, j), (i{+}1, j) \rangle \mid i, j \in \N \}$.
\end{itemize}
We say that a structure $\fA$ is a \textit{$\NNs$-grid} if $\fA$ restricted to elements $G^\fA$ and the signature $\{ G, H, V \}$ is isomorphic to the canonical $\NNs$-grid.
More leniently, $\fA$ is \textit{grid-like} if it contains a homomorphic embedding of~$\fG$.
It is well known that the satisfiability problem posed over subclasses of grid-like structures is undecidable for even inexpressive logics such as $\FL^2$.
The following is almost immediate:
\begin{lemma} \label{lemma_undec_tiling_easy_part}
    The satisfiability problem for $\FLC^2_{\text{rev}}$ posed over subclasses of grid-like structures is $\Pi^0_1$-hard.
    The satisfiability problem for $\FLPC^2_{\text{rev}}$ posed over subclasses of \NNs-grids is $\Sigma^1_1$-hard.
\end{lemma}

The lemma above is, of course, the ``easy'' part of a much larger reduction.
The axiomatisation of grid-like structures and $\NNs$-grids is where the expressive power of $\FLC^4_{\text{rev}}$ and $\FLPC^4_{\text{rev}}$ is needed.
Before writing the advertised formulas, we build the motivating structure we will be looking for in three steps. Suppose $\fG$ is the canonical \NNs-grid.
Letting $E_H$, $E_V$ be binary and $O$ be unary predicate symbols we define the \textit{graphed expansion of $\fG$} to be the structure
$\fG^+$ over the domain $\NNs \cup \N$ with the following extensions:
\begin{itemize}
    \item $\fG^+ \restr_{\NN} := \fG$,
    \item $k \not\in G^{\fG^+}$ for each $k \in \N$,
    \item $k \in O^{\fG^+}$ if and only if $k = 0$,
    \item $E_H^{\fG^+} := \bigcup_{i, j \in \N} \{ \langle (i, j), k \rangle \;\vert\; 1 \leq k \leq i \}$, and
    \item $E_V^{\fG^+} := \bigcup_{i, j \in \N} \{ \langle (i, j), k \rangle \;\vert\; 1 \leq k \leq j \}$.
\end{itemize}
Intuitively, $\fG^+$, when restricted to $\NN = G^{\fG^+}$, is the canonical $\NNs$-grid.
Notice that each $(i, j) \in \NN$ has $i$ elements in $\N$ that are $E_H$-successors and $j$ elements that are $E_V$-successors.
In other words, the coordinates of $(i, j)$ are explicitly encoded in $\fG^+$ as the out-degrees of $E_H$ and $E_V$ respectively.
(In the future, we will simply speak of $E_H$- and $E_V$-degree with ``out'' being left implicit).
We invite the reader to regard $\N$ as the set of extra elements which help encode positions of grid elements.
Notice that the singleton $0 \in O^{\fG^+}$ is not featured in any binary relations (most notably, $E_H$ and $E_V$). This is deliberate as it will act as a \textit{spare part}
in the constructions to come.

We now define \textit{the mapped expansion $\fG^\star$ of $\fG^+$}, where $\fG^+$ itself is the graphed expansion of $\fG$.
For this, we introduce quaternary $R_H$, $R_V$, $S_H$, $S_V$ and ternary predicates $C_H$, $C_V$, whilst setting the following extensions:
\begin{itemize}
    \item $R_H^{\fG^\star} := \bigcup_{i, j, i', j' \in \N}^{i \leq i'} \{ \langle k, (i, j), (i', j'), k \rangle \;\vert\; 1 \leq k \leq i \},$
    \item $R_V^{\fG^\star} := \bigcup_{i, j, i', j' \in \N}^{j \leq j'} \{ \langle k, (i, j), (i', j'), k \rangle \;\vert\; 1 \leq k \leq j \},$
    \item $S_H^{\fG^\star} := \bigcup_{i, j, j' \in \N} \{ \langle k-1, (i, j), (i{+}1, j'), k \rangle \;\vert\; 1 \leq k \leq i + 1 \},$
    \item $S_V^{\fG^\star} := \bigcup_{i, j, i' \in \N} \{ \langle k-1, (i, j), (i', j{+}1), k \rangle \;\vert\; 1 \leq k \leq j + 1 \},$
    \item $C_H^{\fG^\star} := \bigcup_{i, j, i', j' \in \N}^{i \leq i'} \{ \langle k, (i', j'), (i, j) \rangle \;\vert\; 1 \leq k \leq i \},$ and
    \item $C_V^{\fG^\star} := \bigcup_{i, j, i', j' \in \N}^{j \leq j'} \{ \langle k, (i', j'), (i, j) \rangle \;\vert\; 1 \leq k \leq j \}$.
\end{itemize}
Recall that each $(i, j) \in \NN$ sends $E_H$-edges to each $k \in [1, i]$.
Fixing some $(i', j') \in \NN$ with $i \leq i'$, we have that $R_H^{\fG^\star}$
injectively maps $E_H$-edges originating from $(i, j)$ to $E_H$-edges of $(i', j')$. On the other hand,
$S_H^{\fG^\star}$ is a bijection between the $E_H$-edges of $(i, j)$ with the spare part $0$ and $E_H$-edges of $(i{+}1, j')$.
The relation $C_H^{\fG^\star}$ simply remembers which $E_H$-edges of $(i', j')$ are mapped to $E_H$-edges of $(i, j)$ via $R_H^\fA$. Relations
$R_V^{\fG^\star}$, $S_V^{\fG^\star}$ and $C_V^{\fG^\star}$ act similarly.

Lastly,
We say that $\fG^\#$ \textit{is the ordered expansion} of $\fG^\star$, where $\fG^\star$ itself is the graphed and mapped expansion of $\fG$,
if the signature contains two additional relations
$\preceq_H$ and $\preceq_V$ which we will define to be total orders over $\NN$.
For motivational purposes, we will forget that grid elements $a$ and $b$ are pairs of natural numbers
and instead focus on the $E_H$- and $E_V$-degrees of the elements.
In $\fG^\#$ we will have that
\begin{itemize}
    \item $a \preceq^{\fG^\#}_H b$ if and only if the $E_H$-degree of $a$ is no more than that of $b$, and
    \item $a \preceq^{\fG^\#}_V b$ if and only if the $E_V$-degree of $a$ is no more than that of $b$.
\end{itemize}
(Again, the $E_H$- and $E_V$-degrees encode the horizontal and vertical positions of the element).

We will now write the sentence $\varphi := \varphi_1 \wedge \cdots \wedge \varphi_{13}$ one conjunct at a time.
At a high level, the conjuncts simply state facts about the graphed mapped and ordered expansion $\fG^\#$ of $\fG$.
In the sequel we will argue that the satisfaction of $\varphi$ by $\fA$ is sufficient to deduce 
that the structure in question is grid-like.
For readability, we will be using variable sequences $x, y, z, w$ and $w, z, y, x$ instead of $x_1, x_2, x_3, x_4$.
Strictly speaking, the formulas to be defined are not (reverse) fluted, but can be made such by moving quantifiers inwards. 

Fix $\fG^\#$ to be as described above.
We first capture some graphed properties.
Recall that in $\fG^\#$ there is a single spare part element $0 \in O^{\fG^\#}$.
Noting that this element is not part of the grid ($0 \not\in G^{\fG^\#}$)
we have that $\fG^\#$ models:
\begin{equation}\tag{$\varphi_1$}\label{eq_FLPCrev_gen_sat_undec_1}
    \exists_{[=1]} x \; O(x) \wedge \forall x \big(O(x) \to \neg G(x) \big)
\end{equation}
Additionally, recall that the spare part $0$ has no incoming edges $E_H$- or $E_V$-edges in $\fG^{\#}$.
Thus, $\fG^\#$ also models:
\begin{equation}\tag{$\varphi_2$}\label{eq_FLPCrev_gen_sat_undec_2}
    \forall x \Big(O(x) \to \forall y \big(\neg E_H(yx) \wedge \neg E_V(yx) \big) \Big)
\end{equation}

Moving to grid elements, we see that there is a single element in $G^{\fA^\#}$ with no $E_H$- or $E_V$-degree.
Thus, $\fG^\#$ is a model of:
\begin{equation}\tag{$\varphi_3$}\label{eq_FLPCrev_gen_sat_undec_3}
    \exists_{[=1]} x \Big( G(x) \wedge \forall y \big(\neg E_H(xy) \wedge \neg E_V(xy) \big) \Big).
\end{equation}
Additionally, each element in $G^{\fG^\#}$ has a single $H$- and $V$-successor.
Thus, $\fG^\#$ models: 
\begin{equation}\tag{$\varphi_4$}\label{eq_FLPCrev_gen_sat_undec_4}
    \forall x \bigg(G(x) \to \Big(\exists_{[=1]} y \big(H(xy) \wedge G(y) \big) \wedge \exists_{[=1]} y \; \big(V(xy) \wedge G(y) \big) \Big) \bigg).
\end{equation}

For the next two conjuncts fix $(i, j) \in G^{\fG^\#}$.
Notice that the $H$-successor $(i{+}1, j)$ has an $E_H$-degree that is larger by 1 when compared
to its predecessor $(i, j)$. We can thus map the $E_H$-edges from
$(i{+}1, j)$ to the set of $E_H$-edges from $(i, j)$ taken together with the spare part element bijectively. This is exactly how the extension to
$S_H$ in $\fG^\#$ is set up. Noting that $V$-successors have analogous properties we conclude that $\fG^\#$ models the following sentences:
    \begin{align}
        &\bigwedge_{X \in\{H,V\}} \hspace{-4mm} \forall  x y z \Big( \big((E_X(yx) \vee O(x)) \wedge X(yz) \big) {\to}
                \exists_{[=1]} w \big(E_X(zw) \wedge S_X(xyzw) \big) \Big), \tag{$\varphi_5$}\label{eq_FLPCrev_gen_sat_undec_5} \\
        &\bigwedge_{X \in\{H,V\}} \hspace{-4mm} \forall wzy \Big( \big(E_X(zw) \wedge X(yz) \big) {\to}
        \exists_{[=1]} x \big((E_X(yx) \vee O(x)) \wedge S_X(xyzw) \big) \Big). \tag{$\varphi_6$}\label{eq_FLPCrev_gen_sat_undec_6}
    \end{align}

Recall that for grid elements $(i, j), (i', j') \in G^{\fG^\#}$ we have $\fG^\# \models (i, j) \preceq_V (i', j')$
if and only if the $E_V$-degree of $(i, j)$ is no more than that of $(i', j')$.
Fixing $(i, j)$ and its $H$-successor $(i{+}1, j)$ we see that $\fG^\# \models (i, j) \preceq_V (i{+}1, j) \wedge (i{+}1, j) \preceq_V (i, j)$.
That is, $(i, j)$ and its $H$-successor have the same $E_V$-degrees. Thus, $\fG^\#$ models:
\begin{equation}\tag{$\varphi_7$}\label{eq_FLPCrev_gen_sat_undec_7}
    \bigwedge_{X,Y\in\{H,V\}}^{X \neq Y} \forall xy \Big( X(xy) \to \big(x \preceq_{Y} y \wedge y \preceq_{Y} x \big) \Big).
\end{equation}

Now, recall that the ordering $\preceq_X$ ($X \in \{H, V\}$) is total on $G^{\fG^\#}$. Thus, $\fG^\#$ models:
\begin{equation}\tag{$\varphi_8$}\label{eq_FLPCrev_gen_sat_undec_8}
    \bigwedge_{X \in\{H,V\}} \forall xy \; \Big( \big( G(x) \wedge G(y) \big) \to \big( x \preceq_X y \vee y \preceq_X x \big) \Big).
\end{equation}

Taking $(i, j), (i', j') \in G^{\fG^\#}$ with $i \leq i'$
recall that the elements have $E_H$-edges
mapped injectively by $R_H^{\fG^\#}$. Thus, $\fG^\# \models (i, j) \preceq_H (i', j')$ if and only if
$R_H^{\fG^\#}$ is an injection between the $E_H$-edges of $(i, j)$ and  that of $(i', j')$.
We capture the ``only-if'' direction of the dependency with the sentences  \ref{eq_FLPCrev_gen_sat_undec_9} and  \ref{eq_FLPCrev_gen_sat_undec_10}, whilst
the ``if'' direction is handled by  \ref{eq_FLPCrev_gen_sat_undec_11} and  \ref{eq_FLPCrev_gen_sat_undec_12}.

Still holding the supposition that $\fG^\# \models (i, j) \preceq_H (i', j')$,
we write a sentence ensuring that each $E_H$-edge from $(i, j)$ is mapped to some single $E_H$-edge from $(i', j')$:
\begin{equation}\tag{$\varphi_9$}\label{eq_FLPCrev_gen_sat_undec_9}
    \bigwedge_{X \in\{H,V\}} \forall x y z \Big( \big( y \preceq_X z \wedge E_X(yx) \big) \to
                \exists_{[=1]} w \big(E_X(zw) \wedge R_X(xyzw) \big) \Big).
\end{equation}
With the next sentence we require that each $E_H$-edge from $(i', j')$ is a witness (in regard to $R_H^{\fG^\#}$) to at most a single $E_H$-edge from $(i, j)$:
\begin{equation}\tag{$\varphi_{10}$}\label{eq_FLPCrev_gen_sat_undec_10}
    \bigwedge_{X \in\{H,V\}} \forall w z y \Big( \big( y \preceq_X z \wedge E_X(zw) \big) \to
                \exists_{[\leq1]} x \big(E_X(yx) \wedge R_X(xyzw) \big) \Big).
\end{equation}
It is easy to verify that this is indeed how $R_H^{\fG^\#}$ is set up.
Noting that $\preceq_V$ and $R_V$ behave symmetrically we conclude $\fG^\# \models \text{\ref{eq_FLPCrev_gen_sat_undec_9}} \wedge \text{\ref{eq_FLPCrev_gen_sat_undec_10}}$.

For the converse direction of the implication take any $(i, j), (i', j') \in G^{\fG^\#}$ and recall that $\langle k, (i', j'), (i, j) \rangle \in C_H^{\fG^\#}$
if and only if there is some $k' \in \N$ for which $\langle k', (i, j), (i', j'), k \rangle \in R_H^{\fG^\#}$.
In other words, $C_H^{\fG^\#}$ remembers which $E_H$-edges from $(i', j')$ are featured in a mapping (by $R_H^{\fG^\#}$) with $E_H$-edges from $(i, j)$.
We axiomatise this relationship as follows:
\begin{equation}\tag{$\varphi_{11}$}\label{eq_FLPCrev_gen_sat_undec_11}
    \bigwedge_{X \in\{H,V\}} \forall wzy \Big( C_X(wzy) \leftrightarrow \exists_{[=1]} x \big( E_X(yx) \wedge R_X(xyzw)  \big)  \Big).
\end{equation}
Utilising $C_H$ we can then test if each $E_H$-edge of $(i', j')$ is mapped to some $E_X$-edge of $(i, j)$
and, if that is indeed the case, require that the grid elements be related via $\preceq_H$ accordingly.
We do just that with $\varphi_{12}$:
\begin{equation}\tag{$\varphi_{12}$}\label{eq_FLPCrev_gen_sat_undec_12}
    \bigwedge_{X = H, V} \forall y z \Big( \big( G(y) \wedge G(z) \wedge \forall w (E_X(zw) \to C_X(wzy)) \big) \to z \preceq_X y \Big).
\end{equation}
Noting that $\preceq_V$, $R_V$ and $C_V$ behave similarly,
we have that $\fG^\# \models \text{\ref{eq_FLPCrev_gen_sat_undec_11}} \wedge \text{\ref{eq_FLPCrev_gen_sat_undec_12}}$.

Lastly, notice that there are no two grid elements that have the same $E_H$- and $E_V$-degrees.
Thus, $\fG^\#$ models the uniqueness requirement as given by $\varphi_{13}$:
\begin{equation}\tag{$\varphi_{13}$}\label{eq_FLPCrev_gen_sat_undec_13}
    \forall y \; \Big( G(y) \to \exists_{[=1]} z \; \big( \bigwedge_{X = H, V} ({y \preceq_X z} \wedge {z \preceq_X y}) \big) \Big).
\end{equation}

Notice that by stepping inside the realm of periodic counting, we may capture the fact that in $\fG^{\#}$
there are no transfinite positions by defining the sentence $\chi$ limiting the $E_H$- and $E_V$-degrees of elements to finite values:
\begin{equation}\tag{$\chi$}\label{eq_FLPCrev_gen_sat_undec_chi}
    \bigwedge_{X=H, V} \forall x \exists_{[0^{+1}]} y \; E_X(xy).
\end{equation}

Recalling that $\varphi := \varphi_1 \wedge \cdots \wedge \varphi_{13}$ we have showed the following:
\begin{lemma}\label{lemma_undec_tiling_hard_part_1}
    The graphed, mapped and ordered expansion of the canonical $\NNs$-grid is a model of $\varphi \wedge \text{\ref{eq_FLPCrev_gen_sat_undec_chi}}$.
\end{lemma}

We proceed with the other direction as follows:
\begin{lemma}\label{lemma_undec_tiling_hard_part_2}
    Suppose $\fA \models \varphi$. Then $\fA$ is a grid-like structure.
    In addition, if $\fA \models \text{\ref{eq_FLPCrev_gen_sat_undec_chi}}$, then $\fA$ is an $\NNs$-grid.
\end{lemma}
\begin{proof}
    Suppose first that $\fA \models \varphi$.
    Notice that by \ref{eq_FLPCrev_gen_sat_undec_1} there is exactly one element that satisfies $O$ in $\fA$ and, by  \ref{eq_FLPCrev_gen_sat_undec_2}, has no incoming $E_H$- and $E_V$-edges.
    This will be our spare part element in the argument to come.
    Now, take any element $a_{0, 0} \in A$ such that $a_{0, 0} \in G^\fA$ (i.e. $a_{0, 0}$ is a grid element) with finite $E_H$- and $E_V$-degree.
    Such an element is guaranteed to exist by  \ref{eq_FLPCrev_gen_sat_undec_3}. Then,  \ref{eq_FLPCrev_gen_sat_undec_4} gives us that $a_{0, 0}$ has an $H$-successor $a_{1, 0}$ and
    a $V$-successor $a_{0, 1}$. Notice that, by  \ref{eq_FLPCrev_gen_sat_undec_5} each $E_H$-edge originating from $a_{0, 0}$ along with the spare part
    is paired with exactly one edge $E_H$-edge from $a_{1, 0}$ in $S_E^\fA$. That is to say, writing
    $U = \{ b \in A \mid a_{0, 0}b \in E_H^{\fA} \text{ or } b \in O^\fA \}$ and $U' = \{ c \in A \mid a_{1, 0}c \in E_H^{\fA} \}$, we have that
    for each $b \in U$ there is exactly one $c \in U'$ such that $ba_{0, 0} a_{1, 0}c \in S_E^\fA$.
    The reverse is established by  \ref{eq_FLPCrev_gen_sat_undec_6}.
    Clearly, there is a bijection between $U$ and $U'$ thus making the $E_H$-degree of $a_{1, 0}$ one greater than that of $a_{0, 0}$.
    By  \ref{eq_FLPCrev_gen_sat_undec_7} we have that $a_{0, 0} \preceq_V^\fA a_{1, 0}$ and $a_{1, 0} \preceq_V^\fA a_{0, 0}$.
    We first fixate on the fact that $a_{0, 0} \preceq_V^\fA a_{1, 0}$.
    Writing $U = \{ b \in A \mid a_{0, 0}b \in E_V^{\fA} \}$ and $U' = \{ c \in A \mid a_{1, 0}c \in E_V^{\fA} \}$
    we have, by  \ref{eq_FLPCrev_gen_sat_undec_9}, that for each $b \in U$ there is exactly one $c \in U'$ such that $ba_{0, 0} a_{1, 0}c \in R_V^\fA$.
    By  \ref{eq_FLPCrev_gen_sat_undec_10}, for each $c \in U'$ there is at most a single $b \in U$ such that $ba_{0, 0} a_{1, 0}c \in R_V^\fA$.
    We may thus regard $R_V^\fA$ as being an injection between $E_V$-edges of $a_{0,0}$ and that of $a_{1,0}$.
    Then, again by  \ref{eq_FLPCrev_gen_sat_undec_9},  \ref{eq_FLPCrev_gen_sat_undec_10} and the fact that $a_{1, 0} \preceq_V^\fA a_{0, 0}$, we have that
    $R_V^\fA$ is an injection between the $E_V$-edges from $a_{1, 0}$ and that of $a_{0, 0}$.
    By the Cantor-Schr\"{o}der-Bernstein Theorem, their $E_V$-degrees are thus equal.
    A symmetric argument holds for the $E_V$- and $E_H$-degree of $a_{0, 1}$.

    Now, let $a_{1, 1}$ and $a'_{1, 1}$ be, respectively, the $V$-sucessor of $a_{1, 0}$ and the $H$-sucessor of $a_{0, 1}$ promised by  \ref{eq_FLPCrev_gen_sat_undec_4}.
    Using the same arguments as in the paragraph above, it is easy to see that the $E_H$-degrees of $a_{1, 1}$ and $a'_{1, 1}$ coincide; and so do the
    $E_V$-degrees. We claim that $a_{1, 1} = a'_{1, 1}$. By  \ref{eq_FLPCrev_gen_sat_undec_13} we need only show that
    $a_{1, 1}$ is equal to $a'_{1, 1}$ with respect to the orderings $\preceq_H^\fA$ and $\preceq_V^\fA$
    as we already have that $a_{1, 1} \preceq_E a_{1,1}$ and $a_{1, 1} \preceq_V a_{1,1}$ by  \ref{eq_FLPCrev_gen_sat_undec_8}.
    Fixating on $E_H$-edges first, we have,
    by  \ref{eq_FLPCrev_gen_sat_undec_8}, that $a_{1, 1}$ and $a'_{1, 1}$ are comparable by $\preceq_E^\fA$ in some way.
    Suppose, without loss of generality, that $a_{1, 1} \preceq_E^\fA a'_{1, 1}$.
    Writing $U = \{ b \in A \mid a_{1, 1}b \in E_H^{\fA} \}$ and $U' = \{ c \in A \mid a'_{1, 1}c \in E_H^{\fA} \}$
    we have, by  \ref{eq_FLPCrev_gen_sat_undec_9}, that for each $b \in U$ there is exactly one $c \in U'$ such that $ba_{1,1}a'_{1,1}c \in R_E^\fA$, and, by  \ref{eq_FLPCrev_gen_sat_undec_10},
    for each $c \in U'$ there is at most one $b \in U$ such that the same holds.
    That is to say, $R_E^\fA$ is an injection between $E_H$-edges originating from $a_{1, 1}$ and $E_H$-edges from $a'_{1,1}$.
    Notice that since $a_{1, 1}$ and $a'_{1,1}$ both have an equal and  finite $E_H$-degree, we can conclude that $R_E^\fA$ is a bijection
    between the edges.
    Using this, we have that $ca'_{1,1}a_{1,1} \in C_H^\fA$ for each $c \in U'$ by  \ref{eq_FLPCrev_gen_sat_undec_11}.
    Clearly, the antecedents of  \ref{eq_FLPCrev_gen_sat_undec_12} are met and thus $a'_{1, 1} \preceq_E^\fA a_{1, 1}$
    as required.
    Repeating the argument for $\preceq_V^\fA$ we indeed have (by  \ref{eq_FLPCrev_gen_sat_undec_13}) that $a_{1, 1} = a'_{1, 1}$ thus closing the grid.

    By repeating the argument above on element in $G^\fA$ with finite $E_H$- and $E_V$-degree we conclude that
    $\fA$ contains a homomorphic embedding of the canonical $\NNs$-grid thus making it grid-like.

    Supposing, in addition, that $\fA \models \text{\ref{eq_FLPCrev_gen_sat_undec_chi}}$ we have that each element in $G^\fA$ has a finite $E_H$- and $E_V$-degree.
    We may thus unambiguously identify these elements as the pair of their $E_H$-degree $i \in \N$ and $E_V$-degree $j \in \N$.
    Hence, the structure $\fA$ restricted to elements in $G^\fA$ and signature $\{ G, H, V \}$
    is isomorphic to the canonical $\NNs$-grid thus making $\fA$ an $\NNs$-grid as required.
\end{proof}

Combining Lemmas~\ref{lemma_undec_tiling_hard_part_1}~and~\ref{lemma_undec_tiling_hard_part_2}
we have that $\varphi$ is a satisfiable $\FLC^4_{\text{rev}}$-sentence modeled exclusively by (some non-empty subclass of) grid-like structures,
whilst $\varphi \wedge \text{\ref{eq_FLPCrev_gen_sat_undec_chi}}$ is a satisfiable $\FLPC^4_{\text{rev}}$-sentence that is modeled only by (some non-empty subclass of) $\NNs$-grids.
Combining the observation above with Lemma~\ref{lemma_undec_tiling_easy_part} we have the following:
\begin{theorem}\label{th_FLCrev_FLPCrev_gen_undec}
    The satisfiability problem for $\FLC^4_{\text{rev}}$ is $\Pi^0_1$-hard.
    The same problem for $\FLPC^4_{\text{rev}}$ is $\Sigma^1_1$-hard.
\end{theorem}

Note that the adjacent fragment with periodic counting is a fragment of the constructive fragment of $\mathcal{L}_{\omega_1, \omega}$,
which has a $\Sigma^1_1$-complete satisfiability problem \cite{HAREL198551, infi_logic}.
We conclude the section by reformulating results of Theorems~\ref{th_FL_rev_fin_undec}~and~\ref{th_FLCrev_FLPCrev_gen_undec}
in terms of the adjacent fragment:
\begin{corollary} \label{cor_AFC_AFPC_undec}
    The finite and general satisfiability problems for the adjacent fragment with counting are, respectively, $\Sigma^0_1$- and $\Pi^0_1$-complete.
    If periodic counting is permitted, then the general satisfiability problem turns to be $\Sigma^1_1$-complete.
\end{corollary}

\section{Discussion} \label{sec:disc}

In this paper we utilised the homogeneity property of satisfiable $\FLPC$-sentences to establish
a decision procedure for the (finite) satisfiability problem of the new language. With this methodology we not only gained a better understanding of
models of fluted formulas, but also managed to establish decidability of (finite) satisfiability using simpler methods when compared to
Presburger quantifiers discussed in previous literature \cite{ph21, ph_lt23}.

Reflecting on global homogeneity we see that, as opposed to local homogeneity, the rewiring in Lemma~\ref{cor:glob_homo} is impacted by the 
presence of the (in)equality atom.
More precisely, the semantics of predicates of arity at most $\ell$ do not interfere in the rewiring for local $\ell$-homogeneity.
Because of this we may establish a more general result for local homogeneity
of fluted sentences with semantic extensions.
Consider a signature that is split into symbols $\sigma^*$ with a fixed interpretation
(e.g. transitive relation, reversed relation, etc.)
and standard predicate symbols $\sigma$ with no fixed meaning. Furthermore, suppose that the maximum arity of any symbol in $\sigma^*$ is at most $k$.
Then, Lemma \ref{cor:loc_homo} implies the following:
\begin{corollary} \label{cor_last}
    Suppose $\varphi$ is a fluted, normal-form, $(\ell{+}1)$-variable sentence (possibly with periodic counting) over the signature $\sigma \cup \sigma^*$,
    and where $\ell \geq k$. Then if $\varphi$ is satisfiable it is satisfiable in a locally $\ell$-homogeneous model.
\end{corollary}
The same cannot be said about $(\ell{+}1)$-variable sentences when $\ell < k$, and thus establishing an analogue to global homogeneity
(as was done for $\FLPC^2$ with equality in Lemma \ref{cor:glob_homo}) requires case-by-case consideration.
Nonetheless, we believe our approach could not only be used to simplify existing decidability procedures for satisfiability
(e.g. for $\FL$ with a transitive relation and counting \cite{ph_lt23}) but to also expand on expressiveness of fluted languages.

We make use of Corollary \ref{cor_last} when (briefly) analysing the language $\mathcal{L}$
formed by combining $\FLPC$ and $\mathcal{FO}^2_{\text{Pres}}$.
That is to say, $\mathcal{L}$ is $\FLPC$ with $(\ell{+}1)$-atoms $R(x_{\ell{+}1}, R_\ell)$ and $R(x_{\ell{+}1}, x_{\ell{+}1})$ allowed.
The following is almost immediate. (See Appendix \ref{app:combined_logic} for the proof).
\begin{theorem}\label{th:combined_logic}
    The (finite) satisfiability problem for $\mathcal{L}$
    is decidable.
\end{theorem}

Combining the result above with Corollary~\ref{cor_AFC_AFPC_undec},
and noting that more expressive counting quantifiers render the two-variable fragment undecidable
\cite{gradel_undec, bartosz_percent}, one can argue that the language $\mathcal{L}$ is on the edge of decidability (for satisfiability).
There are, however, other maximal fragments with counting that have a decidable satisfiability problem.
The reader might have noticed that, in Section~\ref{sec_undec}, we did not consider the general satisfiability problem for
the 3-variable adjacent fragment with counting (whilst noting that the finite variant is undecidable in Theorem~\ref{th_FL_rev_fin_undec}).
Surprisingly, the problem was recently shown to be in $\Delta^0_1$ by the current author.
The details, however, are beyond the scope of this article.
We conclude the paper by outlining the following problems which, to the best of our knowledge, are open.
\begin{enumerate}
    \item What is the complexity of satisfiability for the 3-variable adjacent fragment with counting?
    \item Is the satisfiability problem for the adjacent fragment with periodic counting $\Sigma^1_1$-hard?
    \item Is the satisfiability problem for the guarded adjacent fragment with counting decidable?
\end{enumerate}

\bibliography{references}

\appendix

\section{Preliminaries}
\label{app:nmf}
\newtheorem*{Restate_nmf}{Lemma~\ref{lma:nmf}}
\begin{Restate_nmf}
    Suppose $\varphi$ is an $\FLPC^{\ell{+}1}$-sentence. Then, we may compute, in polynomial time, an equisatisfiable
    normal-form $\FLPC^{\ell{+}1}$-sentence $\psi$.
\end{Restate_nmf}
\begin{proof}
    We start by assuming that $\varphi$ contains no universal quantifiers. No loss of generality follows this supposition as
    every formula of the form $\forall \theta$ is equivalent to $\exists_{[0]} \neg \theta$.
    Writing $\varphi_0 := \varphi$, take any subformula $\theta := \exists_{[n^{+p}]} \chi$ of $\varphi_0$, where
    $\chi$ is quantifier-free. Supposing there are $k$ free variables in $\theta$, let $q$ be a fresh predicate of arity $k$.
    We write $\psi_1$ as
    $
    \forall^\ell ( q \to \exists_{[n^{+p}]} \chi) \wedge \forall^\ell ( \neg q \to \neg \exists_{[n^{+p}]} \chi)
    $
    and define $\varphi_1$ to be $\varphi_0$ but with $\theta$ replaced by $q$. Clearly, $\varphi_1 \wedge \psi_1 \models \varphi_0$.
    Conversely, if $\fA \models \varphi_0$, we may expand $\fA$ to $\fA'$ by setting $\bar{a} \in q^{\fA'}$ if $\fA, \bar{a} \models \theta$
    for each $\bar{a} \in A^k$. Then, $\fA' \models \varphi_1 \wedge \psi_1$ as required.
    Processing $\varphi_1$ and subsequent sentences in the same way, we are left with a sentence $\varphi_m$ composed solely of proposition letters and
    sentences $\psi_1, \dots, \psi_m$. The conjunction of the aforementioned sentences is then (after rearrangement) of the required form.
\end{proof}

\section{Counting With Reversed Relations}

\label{app:undec_missing}

\begin{claim*}
    Suppose $\mathcal{E}$ is an instance of Hilbert's 10th problem and $\varphi$ is computed
    as described as above Theorem~\ref{th_FL_rev_fin_undec}. Then, $\varphi$ is finitely satisfiable if and only if
    $\mathcal{E}$ has a solution over $\mathbb{N}$.
\end{claim*}
\begin{claimproof}
    Suppose $\fA \models \varphi$.
        We claim that $\{ u \mapsto |A_u^\fA| \mid u \in \text{vars}(\mathcal{E}) \}$ is a satisfying assignment for $\mathcal{E}$.
        Thus, again taking $e \in \mathcal{E}$ in turn, we have that if $e$ is of the form \ref{hilb_1} $u = 1$, then
        $\fA \models \varphi_e \implies |A_u^\fA| = 1$. If $e$ takes the form \ref{hilb_2} $u + v = w$, we then claim
        that $R_e^\fA$ is a bijection between $A_u^\fA \cup A_v^\fA$ and $A_w^\fA$ having cardinality $|A_u^\fA| + |A_v^\fA|$.
        Indeed, by the first conjunct of $\varphi_e$ we have that each element in $A_u^\fA \cup A_v^\fA$ is paired with a single element in $A_w^\fA$;
        the converse is establishes by the second conjunct. Thus, clearly, $|A_u^\fA \cup A_v^\fA| = |A_w^\fA|$.
        But since $\fA \models \psi$ we have that $A_u^\fA \cap A_v^\fA = \emptyset$. This coupled with our initial assumption that $u \neq v$ gives us
        $|A_u^\fA \cup A_v^\fA| = |A_u^\fA| + |A_v^\fA|$ as required.
        Lastly, suppose $e$ takes the form \ref{hilb_3} $u \cdot v = w$. By the first conjunct of $\varphi_e$
        for each $ab \in A_u^\fA \times A_v^\fA$ there is a single $c \in A_w^\fA$ such that $abc \in P_e^\fA$.
        Hence, $P_e^\fA$ gives rise to a function
        $
            f := \{ ab \mapsto c \mid abc \in P_e^\fA \cap (A_u^\fA \times A_v^\fA \times A_w^\fA) \}.
        $
        Thus, writing $f(xy)=z$ in place of $P_e(xyz)$, we claim that $f$ is a bijection between $A_u^\fA \times A_v^\fA$
        and $A_w^\fA$ and has cardinality $|A_u^\fA| \cdot |A_v^\fA|$.
        It is easily seen that $f$ is surjective from either the second or third conjunct of $\varphi_e$.
        To establish injectivity suppose $f(ab) = f(a'b') = c$.
        But, by the second conjunct of $\varphi_e$, we have that
        $
            \fA, c \models \exists_{[=1]} y \; \Big( A_v(y) \wedge \exists x \; \big( A_u(x) \wedge f(xy)=z \big) \Big),
        $
        thus $b = b'$. Notice that this establishes that $b$ is the only element in $A_v^\fA$ for which, under the assignment $c \mapsto z$, $b \mapsto y$,
        $
            \fA, cb \models \exists x \; \big( A_u(x) \wedge f(xy)=z \big).
        $
        Combining this fact with the third conjunct of $\varphi_e$ we have that, again under the assignment $c \mapsto z$, $b \mapsto y$,
        $
            \fA, cb \models \exists_{[=1]} x \; \big( A_u(x) \wedge f(xy)=z \big)
        $
        thus making $a = a'$.
        We finish our argument by noting that $|A_u^\fA| \cdot |A_v^\fA| = |A_u^\fA \times A_v^\fA| = |f| = |A_w^\fA|$ as required.
    
        Conversely, suppose $\mathcal{E}$ has a solution. We define $\pi: \text{vars}(\mathcal{E}) \to \mathbb{N}$ to be the satisfying assignment for $\mathcal{E}$
        and construct a model $\fA$ of $\varphi$ as follows. For each variable $u \in \text{vars}(\mathcal{E})$ set $A_u^\fA$
        be a set of $\pi(u)$ distinct elements and set the domain of $\fA$ to be $A = \bigcup_{u \in \text{vars}(\mathcal{E})} A_u^\fA$
        where each $A_u^\fA \cap A_v^\fA = \emptyset$ for $u \neq v$. Clearly, $\fA \models \psi$.
        If $\varphi_e$ was constructed from $e \in \mathcal{E}$ of the form \ref{hilb_1} $u = 1$, then $A_u^\fA = 1$ as required by $\varphi_e$.
        On the other hand, if $e$ is of the form \ref{hilb_2} $u + v = w$, we have that $\fA \models \varphi_e$ by setting $R_e^\fA$
        to be a bijection between $A_u^\fA \cup A_v^\fA$ and $A_w^\fA$ (this can be done as $\pi(u) + \pi(v) = \pi(w)$ and, by initial assumption, $u \neq v$).
        Lastly, if $e$ is \ref{hilb_3} $u \cdot v = w$, then $\pi(u) \cdot \pi(v) = \pi(w)$. Thus, index elements of $A_u^\fA$ as $a_1 \dots a_{\pi(u)}$,
        elements of $A_v^\fA$ as $b_1 \dots b_{\pi(v)}$ and elements of $A_w^\fA$ as $(c_{i, j})^{1 \leq i \leq \pi(u)}_{1 \leq j \leq \pi(v)}$.
        Clearly, by setting $P_e^{\fA} = \{ a_i b_j c_{i,j} \mid 1 \leq i \leq \pi(u), 1 \leq j \leq \pi(v) \}$ we have that $\fA \models \varphi_e$
        thus concluding the proof.
\end{claimproof}
\section{Discussion}

\label{app:combined_logic}
\newtheorem*{Restate_combined_logic}{Theorem~\ref{th:combined_logic}}
\begin{Restate_combined_logic}
    The (finite) satisfiability problem for $\mathcal{L}$ (i.e. the combination of $\FLPC$ and $\mathcal{FO}^2_{\text{Pres}}$)
    is decidable.
\end{Restate_combined_logic}
\begin{proof}
    Let $\mathcal{L}^{\ell+1}$ be the ${\ell+1}$-variable sub-fragment of $\mathcal{L}$.
    Taking some sentence $\varphi \in \mathcal{L}^{\ell+1}$ we proceed by induction on the number of variables.
    If $\ell+1 = 2$, then $\varphi$ is a sentence in the two-variable fragment with periodic counting which is known
    to have a decidable (finite) satisfiability problem \cite{BenediktKT20}.
    Now, set $\ell+1 > 2$ and suppose that the (finite) satisfiability problem for $\mathcal{L}^\ell$ is decidable.
    We may assume (by allowing $\alpha_r, \beta_t, \gamma_r, \delta_t$ to contain binary predicates of the form $R(x_{\ell+1}, x_{\ell})$ and $R(x_{\ell+1}, x_{\ell+1})$)
    that $\varphi$ is in normal-form (\ref{eq:nmf}).
    Defining $\sigma^*$ to be the set of predicates of arity no more than 2, we have that if $\varphi$ is satisfiable
    then, by Corollary \ref{cor_last}, it is satisfiable in an $\ell$-homogeneous model.
    Applying the variable reduction outlined in Lemmas~\ref{FLPC_var_reduct_1}~and~\ref{FLPC_var_reduct_2} we will then have the required result.
\end{proof}

\end{document}